\documentclass{article}
\usepackage[utf8]{inputenc}
\usepackage[english]{babel}
\usepackage[T1]{fontenc}
\usepackage{lmodern}
\usepackage[a4paper]{geometry}
\usepackage{graphicx}
\usepackage{onlyamsmath}
\usepackage{amsfonts}
\usepackage{amssymb}
\usepackage{ntheorem}
\usepackage{xcolor}

\newcommand{\eps}{\varepsilon}
\newcommand{\seps}{\sqrt\varepsilon}

\newcommand{\E}{\mathbb{E}}

\newcommand{\Pro}{\mathbb{P}}
\newcommand{\R}{\mathbb{R}}

\newcommand{\calA}{\mathcal{A}}
\newcommand{\calB}{\mathcal{B}}
\newcommand{\calC}{\mathcal{C}}
\newcommand{\calE}{\mathcal{E}}
\newcommand{\calF}{\mathcal{F}}
\newcommand{\calI}{\mathcal{I}}
\newcommand{\calK}{\mathcal{K}}
\newcommand{\calS}{\mathcal{S}}

\newcommand{\calH}{\mathcal{H}}
\newcommand{\calL}{\mathcal{L}}
\newcommand{\calM}{\mathcal{M}}

\newcommand{\calR}{\mathcal{R}}
\newcommand{\calX}{\mathcal{X}}

\newcommand{\calO}{\mathcal{O}}

\newcommand{\fg}{\mathfrak{g}}

\newcommand{\1}{\mathbf{1}}

\newcommand{\bfe}{\mathbf{e}}
\newcommand{\bx}{\mathbf{x}}
\newcommand{\by}{\mathbf{y}}

\newcommand{\bE}{\mathbf{E}}
\newcommand{\bF}{\mathbf{F}}
\newcommand{\bG}{\mathbf{G}}
\newcommand{\bH}{\mathbf{H}}

\newcommand{\bL}{\mathbf{L}}

\newcommand{\bP}{\mathbf{P}}

\newtheorem{theorem}{Theorem}[section]

\newtheorem{proposition}{Proposition}[section]

\newtheorem{lemma}{Lemma}[section]

\theoremstyle{nonumberplain}
\theorembodyfont{\upshape}
\newtheorem{proof}{Proof}

\title{An effective fractional paraxial wave equation for wave-fronts in randomly layered media with long-range correlations}

\author{Christophe Gomez \thanks{Aix Marseille Univ, CNRS, I2M, Marseille, France, christophe.gomez@univ-amu.fr,
\newline 
Institut de Math\'ematiques de Marseille (I2M) CMI - Technop\^ole Ch\^ateau-Gombert
39, rue F. Joliot Curie
13453 Marseille Cedex 13 
France \newline
} }

\begin{document}

\maketitle

\abstract{
This work concerns the asymptotic analysis of high-frequency wave propagation in randomly layered media with fast variations and long-range correlations. The analysis takes place in the 3D physical space and weak-coupling regime. The role played by the slow decay of the correlations on a propagating pulse is two fold. First we observe a random travel time characterized by a fractional Brownian motion that appears to have a standard deviation larger than the pulse width, which is in contrast with the standard O'Doherty-Anstey theory for random propagation media with mixing properties. Second, a deterministic pulse deformation is described as the solution of a paraxial wave equation involving a pseudo-differential operator. This operator is characterized by the autocorrelation function of the medium fluctuations. In case of fluctuations with long-range correlations this operator is close to a fractional Weyl derivative whose order, between 2 and 3, depends on the power decay of the autocorrelation function. In the frequency domain, the pseudo-differential operator exhibits a frequency-dependent power-law attenuation with exponent corresponding to the order of the fractional derivative, and a frequency-dependent phase modulation, both ensuring the causality of the limiting paraxial wave equation as well as the Kramers-Kronig relations. The mathematical analysis is based on an approximation-diffusion theorem for random ordinary differential equations with long-range correlations.
}

\begin{flushleft}
\end{flushleft}

\section{Introduction}

In contexts such as geophysics, laser beam propagation through the atmosphere, or medical imaging for instance, frequency-dependent attenuation has been observed at rate proportional to
\begin{equation}\label{eq:powerdecay}
|\omega|^{\lambda}\qquad\lambda\in(0,2),
\end{equation}
for a given $\omega$ representing the angular frequency \cite{fannjiang, gargett, gomez00, holm, kiss, nicolle, sinkus}. Depending on the field of application, such power-law attenuation can be referred to as anomalous diffusion, nonexponential relaxation, inelastic damping, hysteric damping, singular hereditary, or singular memory media \cite{chen}. Accurate wave-propagation models with a power-law attenuation is therefore of great importance for applications in imaging and inverse problems (see \cite{nasholm2013} for a survey focusing on medical applications and references therein). To reproduce such power-law decay, several models have been proposed involving fractional derivatives (see \cite{caputo, gelinsky, hanyga, holm, nasholm2013, ren, sushilov, szabo} for instance). A particular attention has been paid to space-time models. In contrast to frequency-domain models, time-domain models allow numerical simulation of a large variety of boundary value problems \cite{hanyga2002}, and can be easier to implement and less costly \cite{wismer}. 

Another approach, attracting more attention recently, consists in considering the propagation media as random and exhibiting fractal correlation structures or long-range dependencies (see \cite{garnier1, garnier2, solna} and \cite[Chapter 9]{holm}). In other words, the correlation function of the medium fluctuations decays slowly enough to not be integrable at infinity. Experimental measurements in real environments have exhibited long-range correlation properties in different contexts, as in geophysics \cite{chen, dolan} or laser beam propagation through the atmosphere \cite{gargett, sidi}. Another aspect of wave propagation that we address in this paper concerns the paraxial approximation. This approximation consists in describing the wave propagation along a privileged axis, and has been extensively studied and used in applications (see \cite{bailly, bamberger, collins, garnier0, gomez1, trappert} for instance). Under suitable assumptions on physical parameters, this approximation can greatly simplify the description of propagation phenomena as well as their numerical simulations \cite{martin}.   

The aim of this paper is to provide a mathematical derivation, from first principles of physics, of a paraxial wave equation exhibiting a power-law attenuation \eqref{eq:powerdecay} with $\lambda \in (1,2]$, in the context of random propagation media with long-range correlations. Our analysis considers both aspects, paraxial approximation and random medium fluctuations with long-range correlations, under the same limit. The case of mixing random medium fluctuations is also treated for comparison. In homogeneous propagation media the derivation of the paraxial approximation is relatively straightforward, but it becomes much more complex when waves are propagating in heterogeneous media. In applications, wave frequencies are generally sufficiently high so that the interactions between the waves and the fine structures of the medium fluctuations cannot be ignored. Rigorous derivations of the paraxial wave equation in random media, with mixing properties, can be found in \cite{bailly, garnier0} for instance, and in case of long-range correlations in \cite{gomez1}. However, in this latter work the scaling regime is not the same as the one proposed in this paper. As waves propagate over large distances, it is natural to expect some \emph{universal behavior} to describe the statistical properties of the multiple-scattering effects. In case of long-range correlations this refers to the non-central limit theorem \cite{taqqu}, as opposed to the standard central limit theorem for mixing propagation media. The scaling of the non-central limit theorem has been used to study wave propagation in 1D propagation media with long-range correlations through the use of the rough-path theory \cite{marty1, marty2} under general assumptions on the random fluctuations, and random waveguides with a moment technique \cite{gomez2}. In this context, the effects on the propagating wave can be described as a random travel-time shift driven by a fractional Brownian motion, leading to anomalous diffusion phenomena but no power-law attenuation in frequency. Here, we rather consider a central limit theorem scaling in the context of long-range correlations leading to mathematical challenges. This situation has been considered for 1D propagation media \cite{garnier1, garnier2}, but their approach, allowing to exhibit a power-law attenuation of the form \eqref{eq:powerdecay} as well as a fractional derivative in the effective wave equation, does not seem to apply in a 3D setting. Despite more restrictive assumptions on the random fluctuations than in \cite{garnier1, garnier2, marty1, marty2}, the approach we propose can also be applied to more general 3D settings with non-layered fluctuations, but with additional technical difficulties, and will be the aim of future works. As for random media with mixing fluctuations, more general fluctuation models with long-range correlations should not change the overall results as the asymptotic equations and scattering coefficients depend only on the correlation functions of the fluctuation models, and not their precise definitions. 

As already pointed out in the context of the random Schr\"odinger equation with long-range correlations \cite{gomez0}, the central limit theorem scaling can be seen as propagating the non-central limit scaling over longer propagation distances. This latter scaling already producing an effective phase modulation driven by a fractional Brownian motion \cite{bal1}, the wave starts to oscillate very fast over larger propagation  distances, and then have to be treated properly to still exhibit effective nontrivial effects. For the random Schr\"odinger equation, the Wigner transform is used to study the energy propagation by looking at correlations of the wave function, which naturally cancels out the rapid phases, and provides an effective description of the energy propagation through a radiative transfer equation \cite{gomez}. For classical wave propagation problems, an equivalent approach consists in looking at the wave-front along a proper random characteristic time-frame. As a result, the rapid phases still have some effects by averaging the stochasticity to obtain a deterministic spreading for the wave-front. It is worthnoticing that under the central limit theorem scaling, but with long-range correlations, the random travel time has a standard deviation very large compared to the pulse width \cite{garnier1, garnier2}. This is in contrast with the standard O'Doherty-Anstey (ODA) theory with mixing medium fluctuations for which the standard deviation of the random travel time and the pulse width are of the same order (see \cite[Chapter 8]{fouque} and references therein). This \emph{unstable} behavior of the random travel time may have a dramatic effect for applications in inverse problems based on travel time estimations, and a deeper understanding of the propagating wave is required. In the context of a randomly layered media, the pulse deformation can be approximately characterized by a deterministic paraxial wave equation of the form
\[
\partial^2_{tz} \psi + \frac{c_0}{2} \Delta_{\bx} \psi - r_0 D^{2+\gamma}_t \psi = 0 \qquad \gamma\in(0,1),
\] 
where the $z$-variable corresponds to the main propagation axis, the $\bx$-variable to the transverse section (see Figure \ref{fig1}), $t$ to the time variable, $c_0$ to the background wave speed, and $r_0>0$ is a constant. Also, $D^{2+\gamma}_t$ stands for the Weyl fractional derivative with respect to time and order between $2$ and $3$ (see \eqref{def:weyl_dev}), depending on the power decay rate $\gamma\in(0,1)$ of the correlation function of the medium fluctuations. This fractional derivative ensures the causality of the paraxial wave equation in the sense that for a given time $t$ the equation involves only the prior knowledge, $\psi(\tau)$ for $\tau \leq t$, of the wave $\psi$. In the Fourier domain, this equation can be recast as a Schr\"odinger equation of the form
\[i\frac{\omega}{c_0} \partial_z \check \psi + \frac{1}{2}\Delta_\bx \check \psi + \tilde r_0 \omega|\omega|^{1+\gamma} \check \psi = 0, \qquad \text{with} \qquad \check \psi(\omega,\bx,z)=\int \psi(t,\bx,z)e^{i\omega t} dt,\]
where $\tilde r_0$ is a constant with positive imaginary part. For $\gamma=1$ the above fractional derivative turns to a classical third order derivative. This equation provides a frequency-dependent power-law attenuation of the form \eqref{eq:powerdecay} with $\lambda=1+\gamma\in(1, 2]$. This range of values for $\lambda$ is typical of attenuation in biological tissues \cite{holm2010}.

The analysis developed in this paper relies on an approximation-diffusion theorem, which is usually used for mixing fluctuations. Despite some restrictions on the noise model for long-range correlations, the proof of this theorem requires a very careful attention due to the nonintegrability of the correlation function at infinity. Also, our restriction does not allow us to capture power-law attenuation with  $\lambda\in(0,1)$, and further investigations would be required to capture these cases.

The paper is organized as follows. In Section \ref{sec:model}, we describe the physical model under consideration and introduce the main assumptions we need to derive the limiting fractional paraxial wave equation, which is presented and discussed in Section 3. Section 4 presents the influence of the slow decay of the random medium correlations on the wave-front travel time. Section \ref{sec:modal_dec} reformulates the propagation problem in terms of a proper random ordinary differential equation, and Section \ref{sec:propagator} describes the asymptotic behavior of this equation. Sections \ref{proof:th_main}, \ref{proof:th_asympt}, and \ref{proof:th_asymp} are dedicated to the proof of the main results.

\section{The wave model}\label{sec:model}

In this section, we describe the physical model under consideration, the random medium fluctuations, and introduce some assumptions that are necessary to derive the main results stated in Section \ref{sec:main_results}.

\subsection{Random wave equation}

In this paper a three-dimensional linear wave equation model is considered
\[
\Delta p - \frac{1}{c^2(z)} \partial^2_t p  = \textbf{F}(t,\bx,z) \qquad (t,\bx,z)\in \R\times \R ^2 \times \R,
\]
equipped with null initial conditions
\begin{equation}\label{eq:cond_init}
 p(t=0,\bx, z) = \partial_t p(t=0,\bx, z) = 0\qquad (\bx,z)\in \R ^2 \times \R,
\end{equation}
meaning that the system is initially at rest. The $z$-coordinate represents the main propagation axis, while the $\bx$-coordinate represents the transverse section (cf. Figure \ref{fig1} for an illustration of the setting). 
\begin{figure}
\begin{center}
\includegraphics[scale=0.25]{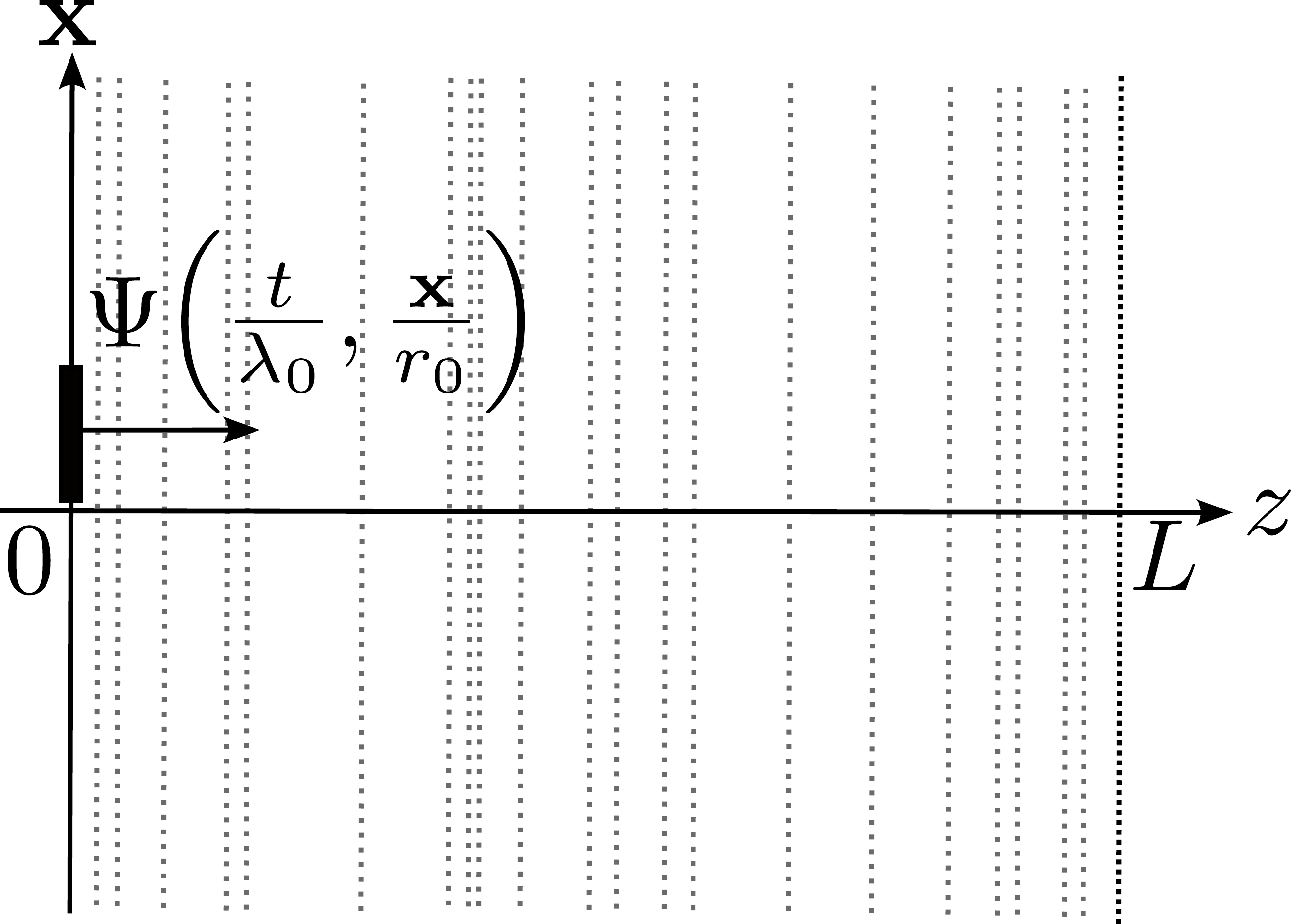}
\end{center}
\caption{\label{fig1} Illustration of the physical setting.}
\end{figure}
The forcing term $\textbf{F} (t,\bx,z)$ is assumed to be of the form
\[ 
\textbf{F} (t,\bx,z):= \Psi\Big(\frac{t}{\lambda_0},\frac{\bx}{r_0}\Big)\delta(z),
\]
and represents a source emitting a pulse located in the plan $z=0$, with central wavelength $\lambda_0$, and beam radius $r_0$. To be consistent with \eqref{eq:cond_init} the source profile $\Psi$ needs to be supported in time in $(0,\infty)$. For mathematical convenience, we also assume that the frequency $\omega=0$ is not supported by the Fourier transform in time of $\Psi$. This assumption will be made more precise later in the paper. Finally, the wave-speed profile is assumed to be of the form
\begin{equation}\label{def:c}
\frac{1}{c^2(z)} := \frac{1}{c^2_0} \Big(1 + \nu\Big(\frac{z}{l_c}\Big)\mathbf{1}_{(0, L)}(z)\Big),
\end{equation}
where $c_0$ stands for the background wave speed, $\nu$ represents the fluctuations of the wave speed, and $l_c$ is the correlation length of these fluctuations. In other words, the correlation length can be seen as the typical scale of variation for the random fluctuations. The wave speed varies only in one direction, here the $z$-direction, providing a \emph{layered} propagation medium. In practice, the wave-speed variations are almost impossible to determine exactly, and it is therefore reasonable to consider these fluctuations as random. The indicator function in \eqref{def:c} indicates that the random perturbations take place only in the slab $z\in(0,L)$. 

Our goal is to provide a description of the wave at $z=L$ through a paraxial wave equation. This wave is referred to as the \emph{transmitted wave}, in contrast with the one observed at $z=0$, which is referred to as the \emph{reflected wave}.

\subsection{The random fluctuations}

In this paper, the random fluctuations are assumed to be given by
\[\nu(z) := \Theta(\sigma V(z)),\]
where $\Theta$ is an odd smooth bounded function with 
\[
\theta'_0:=\Theta'(0)\neq 0 \qquad\text{and}\qquad \sup |\Theta| < 1,
\]
and $V$ is a mean-zero stationary Gaussian random process. Here, the function $\Theta$ plays no significant role. Gaussian random processes being not bounded, this function just guarantees, for a modelization purpose, that $c^2(z)$ is actually always positive (see \eqref{def:c}). Also, for $\sigma$ small enough, the Taylor expansion
\[
\Theta(\sigma V(z)) = \sigma V(z) \Theta'(0) +\calO(\sigma^3)
\]
($\Theta''(0)=0$ since $\Theta$ is odd) indicates that the medium fluctuations are driven by $V$.

To derive the fractional behavior of the paraxial wave equation, or equivalently the frequency-dependent power-law attenuation, we consider random fluctuations with \emph{long-range correlations}, or in other words with \emph{slowly decaying correlations}. Behind this terminology, we assume that the two-point correlation function is not integrable:
\[
\int_0^\infty |\E[V(z + s)V(z)]| \, ds = \infty.
\]
\begin{figure}
\begin{center}
\includegraphics[scale=0.3]{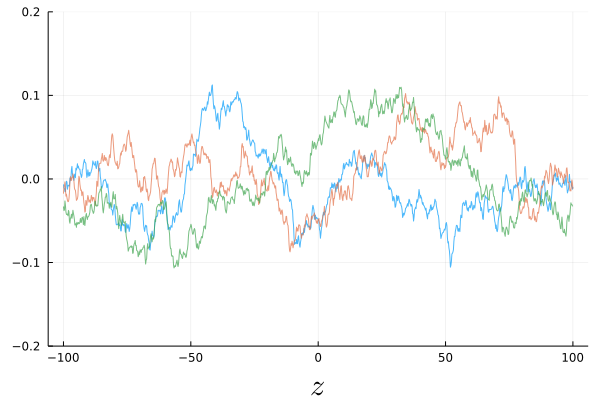}
\includegraphics[scale=0.3]{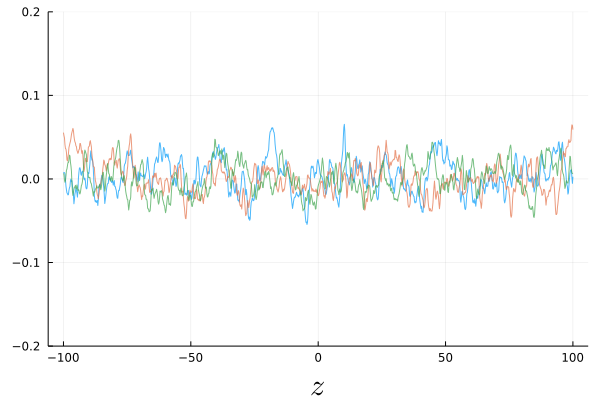}
\end{center}
\caption{\label{fig_LR} Illustration of three realizations of the random process defined by \eqref{def:V} with long-range correlations for the left-picture, and short-range correlations for the right-picture. Here, $\mu=1$, $\alpha=1/4$, $a(p)=\1_{(-10,10)}(p)$, $\beta=1/2$ (that is $\gamma=1/2$) for the left-picture illustrating the long-range correlations, and $\beta=1/6$ (that is $\gamma=3/2$) for the right-picture illustrating short-range correlations.}
\end{figure}
In Figure \ref{fig_LR}, we illustrate the difference of statistical behavior between long-range correlations and \emph{short-range correlations} (or in other words with \emph{rapidly decaying correlations}), the latter having an integrable two-point correlation function. From these pictures, one can observe that slowly decaying correlations produce longer excursions of the random trajectories, due to the persistence of the correlations, than for rapidly decaying correlations. In this latter case, the trajectories cannot really produce correlation patterns and look almost like the ones of a white-noise. These simulations are random trajectories of a stochastic process we now introduce precisely.  

The method we propose to analyze the transmitted and reflected waves is based on the perturbed-test-function method \cite{kushner, gomez3}, and to use this technique under the context of long-range correlations we need some assumptions on $V$. The following construction has already been used to study the impact of random fluctuations with long-range correlations on the Schr\"odinger equation \cite{bal1, gomez, gomez0} and nonlinear oscillators \cite{gomez3}. We consider $V$ as being a linear superposition of Ornstein-Uhlenbeck type processes by setting
\begin{equation}\label{def:V}
V(z) := \int_{-\infty}^z \int_S e^{- \mu |p|^{2\beta} (z-u) } B(du, dp)=  \int_S V(z,dp),
\end{equation}
with
\begin{equation}\label{def:Vdp}
V(z,dp):= \int_{-\infty}^z e^{- \mu |p|^{2\beta} (z-u) } B(du, dp),
\end{equation}
where $\mu, \beta>0$, and $S=(-r_S, r_S)\subset \R$ is a bounded interval containing $0$ for $r_S >0$. Also, $B$ is a Gaussian random measure with covariance function
\[
\E[B(du, dp)B(dv, dq)] := 2\mu \,r(p) |p|^{2\beta}\,\delta(u-v)\,\delta(p-q)\,du\,dv\,dp\,dq,
\] 
with
\[
r(p) := \frac{a(p)}{|p|^{2\alpha}},\qquad \alpha<1/2,\qquad\text{and}\qquad p\in S\setminus\{0\},
\]
where $a$ is a nonnegative smooth bounded function with $a(0)>0$. From this definition, we obtain the autocorrelation function for $V$:
\begin{equation}\label{def:R}
R(z) := \E[V(z_0 + z)V(z_0)] = \int_S e^{-\mu|p|^{2\beta} |z| }r(p) dp.
\end{equation}
Note that we need $\alpha<1/2$ for the process $V$ to be well-defined. After some algebra, we can exhibit the asymptotic behavior of the autocorrelation function:
\begin{equation}\label{eq:decay_R}
R(z) \underset{|z|\to \infty}{\sim} \frac{R_0}{|z|^\gamma},
\end{equation}
with 
\begin{equation}\label{def:R0}
R_0 := a(0)\int_{-\infty}^{\infty} \frac{e^{-\mu |p|^{2\beta}}}{|p|^{2\alpha}}dp,\qquad\text{and}\qquad \gamma:=\frac{1-2\alpha}{2\beta}>0.
\end{equation}
We refer to Figure \ref{fig_LR} for illustrations of this random process. An odd bounded function being of Hermite-rank one, we have (see \cite{marty1} for more details)
\[
\E[\nu(z_0 + z)\nu(z_0)] \underset{|z|\to \infty}{\sim} \frac{R_0 \Theta_1}{|z|^\gamma},
\]
with
\[
\Theta_1 := \Big(\frac{1}{\sqrt{2\pi}}\int  \Theta(\sigma u) \, u\, e^{-u^2/2}du\Big)^2.
\]
As a result, depending on the value of $\gamma$, the two-point correlation function for $\nu$ can be integrable or not. In other words, the medium fluctuations exhibit short-range correlations for $\gamma>1$ or long-range correlations for $\gamma \in (0,1]$. 

In this paper, for a mathematical tractability purpose, a specific form for $V$ is considered. More general fluctuation models have been considered in \cite{marty1, marty2}, but the method they propose, based on the rough-path theory, does not seem to apply in the scaling regime described below. This method can be used to analyze the competition between randomness and periodicity in random differential equations emanating from wave propagation problems \cite{marty0}. However, in our scaling regime, the structure of the periodic components turns out to involve the randomness itself in a way that cannot be controlled easily by this strategy. In \cite{garnier1, garnier2}, the authors use also more general fluctuation models in a scaling regime similar to the one presented here. Nevertheless, the method they use to analyze the problem cannot be applied for 3D propagation media, it is designed for 1D propagation media, that is without the transverse variable $\bx$.

\subsection{The scaling regime}

The asymptotic analysis we provide here is based on a separation of the characteristic scales of the problem. The scales of interest are the propagation distance $L$ into the random medium, the central wavelength $\lambda_0$, the correlation length $l_c$ of the medium fluctuations, the beam radius $r_0$, and the fluctuation strength $\sigma$. Our scaling regime is based on the four following assumptions. First, we consider a high-frequency regime, that is the central wavelength is small compared to the propagation distance:
\[\eps:=\frac{\lambda_0}{L} \ll 1.\]
Second, we assume that the correlation length is of order the central wavelength
\[
l_c \sim \lambda_0,
\]
providing a full interaction between the random fluctuations and the propagating wave. Third, we assume that the beam width $r_0$ satisfies 
\[ \frac{r^2_0}{\lambda_0} \sim L,\]
so that the Rayleigh length is of order the propagation distance, which is crucial to obtain the  paraxial approximation. In fact, the Rayleigh length
is defined as the distance from the beam waist to the place where its cross-section is doubled by diffraction, and in homogeneous media it is of order $r_0^2/\lambda_0$. Finally, the strength of the fluctuations is assumed to be small, so that we place ourselves in a weak-coupling regime:
\[\sigma \ll 1.\]
To fix the ideas, we set 
\[L\sim 1,\qquad \lambda_0 = l_c = \eps, \qquad\text{and}\qquad r_0 = \sigma=\seps.\]
The choice of $\sigma$ allows us to derive a nontrivial limit for both short-range and long-range correlations.

Our choice on the parameter scalings leads to the system
\begin{equation}\label{eq:waveeq_eps}
\Delta p_\eps - \frac{1}{c^2_\eps(z)} \partial^2_{tt} p_\eps  = \Psi\Big(\frac{t}{\eps},\frac{\bx}{\seps}\Big)\delta(z) \qquad (t,\bx,z)\in \R\times \R ^2 \times \R,
\end{equation}
with
\[
\frac{1}{c^2_\eps(z)} = \frac{1}{c^2_0} \Big(1 + \nu_\eps\Big(\frac{z}{\eps}\Big) \mathbf{1}_{(0, L)}(z)\Big)\qquad\text{and}\qquad \nu_\eps(z):=\Theta\big(\seps V(z)\big).
\]

\section{The main results}\label{sec:main_results}

To state our main result, we follow the strategy of \cite{garnier1, garnier2}, and introduce the random travel time 
\begin{equation}\label{def:RTT}
T^0_\eps(L):= \frac{L}{c_0} + \frac{1}{2 c_0} \int_0^L  \nu_\eps(z/\eps) dz,
\end{equation} 
corresponding to the expected travel time $L/c_0$ with a random correction, and the wave-front 
\begin{equation}\label{def:p_tr_eps}
p^L_{tr,\eps}(s,\by) := p_\eps\big( T^0_\eps(L) + \eps\, s, \seps\, \by, L \big)\qquad (s,\by)\in\R\times \R^2.
\end{equation}
This wave-front corresponds to the wave observed at the end of the random section ($z=L$), on a time window corresponding to the pulse width $\eps$, and centered at the random travel time $T^0_\eps(L)$.

Before stating our first result, which is proved in Section \ref{proof:th_main}, we introduce some notations. We consider the following Fourier transform convention,
\[
\hat f(\omega,\kappa) := \iint f(s,\by)e^{i \omega (s-\kappa \cdot \by)} ds\, d\by,
\]
and
\[
f(s,\by) := \frac{1}{(2\pi)^3}\iint \hat f(\omega,\kappa)e^{-i \omega (s-\kappa \cdot \by)} \omega^2 d\omega \, d\kappa,
\]
which is convenient to study space-time problems. Denoting 
\[
\calS_{0}(\R\times\R^2)=\Big\{\psi\in \calS(\R\times\R^2):\quad \int \phi(s,\by)ds =0,\quad\forall \by \in\R^2\Big\},
\]
where $\calS(\R\times\R^2)$ stands for the Schwartz class, $\calS'_{0, s,\by}(\R\times\R^2)$ denotes the set of tempered distributions restricted to $\calS_{0}(\R\times\R^2)$ w.r.t. the variables $s$ and $\by$. This restriction to $\calS_{0}(\R\times\R^2)$ is required for the paraxial wave equation \eqref{eq:parax1} to be well-posed. One can remark that our source term $\Psi$ belongs to $\calS_{0}(\R\times\R^2)$, since we assume the frequency $\omega=0$ to not be supported by the source. Below, $\calC^0_z$ (reps. $\calC^1_z$) stands for the set of $\calC^0$-functions (resp. $\calC^1$-functions) w.r.t. the $z$-variable. 

\begin{theorem}\label{th:main}
The family $(p^L_{tr,\eps})_\eps$ converges in probability in $\calC(\R\times \R^2)$ to 
\begin{equation}\label{def:ptr0}
p^L_{tr}(s,\by) = \frac{1}{2}\calK(\cdot,\cdot,L) \ast \Psi(s,\by),\qquad(s,\by)\in \R\times\R^2,
\end{equation}
where, in the Fourier domain,
\begin{equation}\label{def:K}
\hat \calK(\omega,\kappa,z) := e^{-\theta'^2_0 \omega^2 ( \Gamma_c(\omega)+i  \Gamma_s(\omega)) z/(8 c_0^2)}  e^{-i\omega c_0|\kappa|^2 z/2},
\end{equation}
with
\begin{equation}\label{def:Gamma_coef}
\Gamma_c(\omega):=2\int_0^\infty R(s)\cos\Big(\frac{2 \omega s}{c_0}\Big) ds\qquad\text{and}\qquad \Gamma_s(\omega):=2\int_0^\infty R(s)\sin\Big(\frac{2 \omega s}{c_0}\Big) ds.
\end{equation}
Here, $R$ is the correlation function of the medium fluctuations \eqref{def:R}. The convolution kernel $\calK$ is the unique solution in $\calC^0_z([0,\infty),\calS'_{0,s,\by}(\R\times\R^2)) \cap \calC^1_z((0,\infty),\calS'_{0,s,\by}(\R\times\R^2))$ to the paraxial wave equation 
\begin{equation}\label{eq:parax1}
\partial^2_{sz} \calK + \frac{c_0}{2} \Delta_{\by} \calK - \calI(\calK) = 0, 
\end{equation}
with $\calK(s,\by,z=0)=\delta(s)\delta(\by)$, and
\begin{equation}\label{def:I}
\calI (\psi)(s) := \frac{\theta'^2_0}{8c_0^2} \int_{-\infty}^s R\Big(\frac{c_0 (s-\tau)}{2}\Big) \partial^3_{sss} \psi(\tau) d\tau\qquad s\in \R.
\end{equation}
\end{theorem}

The asymptotic transmitted wave-front $p^L_{tr}$, at the end of the random section ($z=L$), can be written in term of a convolution where $\calK$ represents the pulse deformation. From the explicit formulation of $\calK$ in the Fourier domain, the pulse shape is affected in a way which is consistent with the standard ODA theory (see \cite[Chapter 8]{fouque} and references therein for more details) even if we are not considering mixing fluctuations. Typically, according to this theory, the propagating pulse exhibits a deterministic spreading characterized by a frequency-dependent attenuation and phase modulation. In our context, we observe these effects through $\omega^2\Gamma_c(\omega)$ (which is positive thanks to the Bochner theorem) and $\omega^2\Gamma_s(\omega)$ respectively. Here, these two terms are similar to the ones obtained in \cite{garnier1, garnier2}, and are well defined even for slowly decaying correlations thanks to the oscillatory functions. We refer to Figure \ref{fig:pulsespread} for illustrations regarding the influence of the kernel $\calK$ on the propagating pulse spreading.
\begin{figure}
\begin{center}
\includegraphics[scale=0.3]{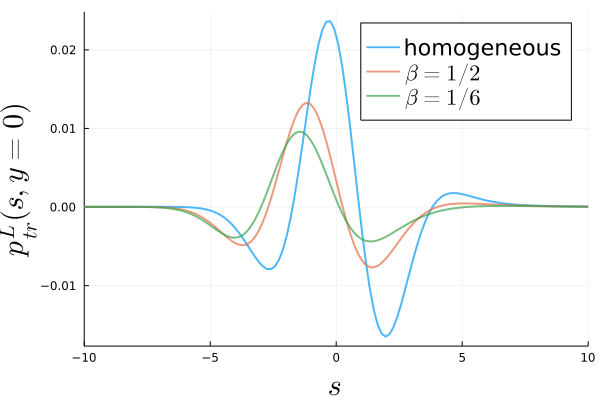}
\includegraphics[scale=0.3]{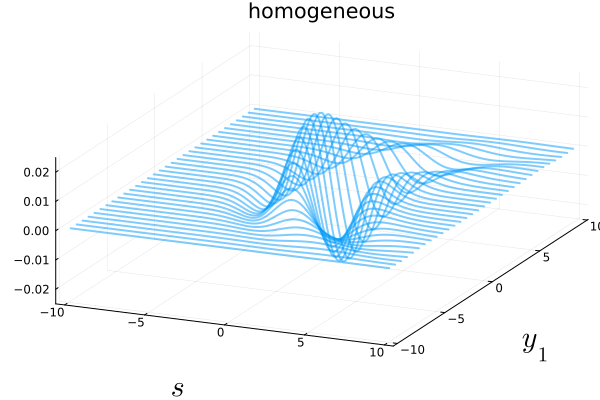}\\

\includegraphics[scale=0.3]{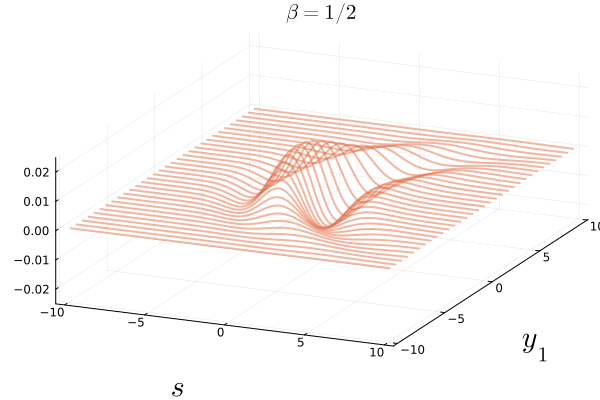}\includegraphics[scale=0.3]{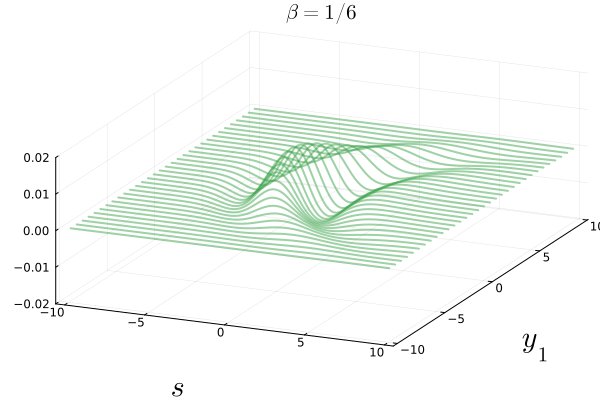}
\end{center}
\caption{\label{fig:pulsespread} Illustrations of the profile $p^L_{tr}(s,y_1,y_2=0)$ in the homogeneous case (blue lines) and for two values of $\beta$ ($\beta=1/2$ for the orange curves and $\beta=1/6$ for the green curves). We take $\alpha=1/4$, $\mu=2$, $a(p)=\1_{(-10,10)}(p)$, $L=5$, $c_0=\theta'_0=1$ 
and a source profile given by $\hat \Psi(\omega,\kappa)=2\omega^2 e^{-\omega^2(1+\kappa^2)}$ centered at $s=0$ in the time domain for simplicity.}
\end{figure}

The ODA theory for mixing fluctuations also provides a random time-shift driven by a standard Brownian motion, meaning that the transmitted pulse exhibits a random arrival time at $z=L$ of order the pulse width. In the context of long-range correlations the situation is more delicate. The aforementioned random time-shift is already compensated in Theorem \ref{th:main} by considering the random travel time $T^0_\eps$ in the definition of the transmitted wave-front \eqref{def:p_tr_eps}. This allows to remove pathological behaviors when studying the asymptotic of the transmitted wave. As described in Section \ref{sec:travel_time}, for rapidly decaying correlations, $T^0_\eps$ can be approximated by a Brownian motion with mean $L/c_0$ and a standard deviation of order the pulse width, which is consistent with the standard ODA theory. In the case of long-range correlations, terms that lead to an effective random time-shift in the context of rapidly decaying correlations would now blow up. As we will see in Section \ref{sec:travel_time}, $T^0_\eps$ can be approximated by a fractional Brownian motion, with a Hurst index ranging from $1/2$ to $1$, and a standard deviation very large w.r.t. the pulse width. This is the reason we compensate this term in \eqref{def:p_tr_eps}, and then avoid this blow up in the derivation of the pulse spreading. These facts will be made more precise in Section \ref{sec:travel_time}, in which we show that $T^0_\eps$ is a convenient approximation of the travel time along random characteristics to reach depth $z=L$
\[
\int_0^L \frac{dz}{c_\eps(z)}. 
\]     

Regarding the backscattered signal at $z=0$, it can be shown that
\[
p_{bk,\eps}(s,\by) := p_\eps\big(\eps s, \seps \by, 0 \big)
\]
converges in probability to $0$, in $\calC(\R\times\R^2)$, as $\eps\to 0$. This is a consequence of Theorem \ref{th:asymp} and it is consistent with \cite[Chapter 9]{fouque}, in which the authors show that the backscatter wave is made of a \emph{small} incoherent signal that can be described through a random field. We will not go in this direction here since it requires involved mathematical developments that are beyond the scope of this paper.

In the time domain, the pulse deformation is described by the paraxial wave equation \eqref{eq:parax1} involving an integral term similar to the one obtained in \cite{garnier1, garnier2}. It is interesting to note that this integro-differential operator preserves the causality, since at a fixed time $s$ it involves only the knowledge of $\calK(\tau,\cdot,\cdot)$ for $\tau\leq s$. Also, in the Fourier domain, \eqref{eq:parax1} can be written as the following Schr\"odinger equation
\begin{equation}\label{eq:schrodinger}
i\frac{\omega}{c_0}\partial_z \check \calK(\omega,\by, z) + \frac{1}{2} \Delta_\by \check\calK(\omega,\by, z)  + i\frac{\theta'^2_0 \omega^3}{8c_0^3}  ( \Gamma_c(\omega)  + i \Gamma_s(\omega) ) \check \calK(\omega,\by) = 0,
\end{equation}
where
\begin{equation}\label{def:Fs}
\check \calK(\omega,\by,z) := \int \calK(s, \by, z)e^{i \omega s} ds.
\end{equation}
Moreover, from \eqref{def:K}, one can see that the effects of the propagation medium in \eqref{eq:schrodinger}, a frequency-dependent attenuation and dispersion, satisfy the Kramers-Kronig relations \cite{kramers, kronig}. These relations are stated more precisely in the following proposition, and proved in Appendix \ref{proof:KK}. Both the causality and the Kramers-Kronig relations contribute to the physical relevance of our paraxial wave equation. These properties are not surprising for \eqref{eq:schrodinger} since this equation is obtained from first principles of physics. 
\begin{proposition}\label{prop:KK}
The effective frequency-dependent attenuation $\omega^2 \Gamma_c(\omega)$ and dispersion $\omega^2 \Gamma_s(\omega)$ are analytic functions w.r.t. $\omega$ on the complex upper half-plane, and satisfy the following Kramers-Kronig relations in $\calS'(\mathbb{R})$,
\[
\calH(\omega'^2 \Gamma_c(\omega') )(\omega) = \omega^2 \Gamma_s(\omega)\qquad \text{and}\qquad \calH(\omega'^2 \Gamma_s(\omega') )(\omega) = - \omega^2 \Gamma_c(\omega),
\]
where $\calH$ stands for the Hilbert transform, and $\calS'(\mathbb{R})$ is the set of tempered distributions on $\R$.
\end{proposition}

In case of long-range correlations, $\calI$ can be approximated by a Weyl fractional derivative whose order depends on the decay rate $\gamma$ of the correlation function $R$ at infinity (see \eqref{eq:decay_R}). Before stating the result, let us briefly introduce the notion of Weyl derivative, which is given for $\gamma\in(0,1)$ by
\[
D^\gamma f(s) := \frac{\gamma}{\Gamma(1-\gamma)} \int_{-\infty}^s \frac{f(s)-f(\tau)}{(s-\tau)^{1+\gamma}} d\tau\qquad s\in\R,
\]
whenever this quantity is well-defined, and $\Gamma$ stands for the Gamma function. For instance, $f$ can be a bounded $\gamma'$-H\"older function with $\gamma<\gamma'$. However, for $\calC^1$-functions with fast enough decay at $-\infty$, the Weyl derivative can be rewritten as
\[
D^\gamma f(s) = \frac{1}{\Gamma(1-\gamma)} \int_{-\infty}^s   \frac{f'(\tau)}{(s-\tau)^{\gamma}} d\tau.
\]   
To define higher order derivatives, one can just set
\begin{equation}\label{def:weyl_dev}
 D^{j+\gamma} f(s) := D^\gamma f^{(j)}(s) = \frac{1}{\Gamma(1-\gamma)} \int_{-\infty}^s   \frac{f^{(j+1)}(\tau)}{(s-\tau)^{\gamma}} d\tau\qquad j \in\mathbb{N}, 
\end{equation}
assuming $f$ smooth enough, with enough decay at $-\infty$ of its derivatives $f^{(j)}$. These latter requirements hold true for the kernel $\calK$ as soon as $z>0$ thanks to the damping term $\omega^2 \Gamma_c(\omega)$ in \eqref{def:K}. Therefore, \eqref{eq:decay_R} and \eqref{def:I} suggests that the integro-differential operator $\calI$ in \eqref{eq:parax1} can be approximated as follows
\[
\calI (\calK) \propto  D^{2+\gamma}_s \calK.
\]
In what follows, we emphasize that the fractional derivative $D^{2+\gamma}$ acts on the $s$-variable with the notation $D^{2+\gamma}_s$. To derive properly this observation we rescale the correlation function as follows. We replace the correlation function $R$ with the following scaled version, 
\begin{equation}\label{def:scaleR}
\sigma(l_0) R(z / l_0),
\end{equation}
where $l_0$ will be sent to $0$, and 
\begin{equation}\label{def:sigmal0}
\sigma(l_0) := \left\{ \begin{array}{ccc}
\displaystyle \frac{1}{l_0^{\gamma}} &\text{if}& \gamma\in (0,1),\\
&&\\
\displaystyle \frac{1}{l_0|\ln(l_0)|} & \text{if} & \gamma=1, \\
&&\\
\displaystyle \frac{1}{l_0} & \text{if} & \gamma >1. \\
\end{array}\right.
\end{equation}
In other words, we assume that the correlation length $l_0$ is small compared to the pulse duration. Under this scaling, the attenuation and dispersion coefficients read
\[
\Gamma_c(\omega,l_0):=2\sigma(l_0)\int_0^\infty \E[V(0)V(s/l_0)]\cos(2 \omega s/c_0) ds,
\]
and
\[ 
\Gamma_s(\omega,l_0):=2\sigma(l_0)\int_0^\infty \E[V(0)V(s/l_0)]\sin(2 \omega s/c_0)ds.
\] 
We can define accordingly, following \eqref{def:ptr0} and \eqref{def:K}, the transmitted wave-front $p^L_{tr,l_0}$ for which we have the following result proved in Section \ref{proof:th_asympt}. 
\begin{theorem}\label{th:asympt}
The family $(p^L_{tr,l_0})_{l_0}$ converges in $\calC(\R\times\R^2)$, as $l_0 \to 0$, to 
\[
p^L_{tr,0}(s,\by) := \frac{1}{2} \calK_0(\cdot,\cdot,L) \ast \Psi(s,\by)\qquad(s,\by)\in \R\times\R^2.
\]
Here, $\calK_0$ is defined in the Fourier domain by
\begin{equation}\label{def:hatK0}
\hat \calK_0(\omega,\kappa,z) :=
\left\{ 
\begin{array}{ccc}
\displaystyle e^{-\theta'^2_0 \omega^2 \Gamma_0 z/(8 c_0^2)}  e^{-i\omega c_0|\kappa|^2 z/2} & \text{if} & \gamma \geq 1, \\
&&\\
&&\\
\displaystyle  e^{-\theta'^2_0 R_0 |\omega|^{1+\gamma} ( \Gamma_{c,0}(\omega)+i  \Gamma_{s,0}(\omega)) z/(8 c_0^2)}  e^{-i\omega c_0|\kappa|^2 z/2} &\text{if} & \gamma \in (0,1),
\end{array}
\right.
\end{equation}
with
\begin{align*}
\Gamma_{c,0}(\omega) &= \Gamma(1-\gamma)\cos\Big(\frac{(1-\gamma)\pi}{2}\Big)\Big(\frac{2}{c_0}\Big)^{\gamma-1}, \\
 \Gamma_{s,0}(\omega) &= \Gamma(1-\gamma)\sin\Big(\frac{(1-\gamma)\pi}{2}\Big)\Big(\frac{2}{c_0}\Big)^{\gamma-1}\emph{sign}(\omega),
\end{align*}
and 
\begin{equation}\label{def:Gamma0}
\Gamma_{0} := \left\{
\begin{array}{ccc}
\displaystyle \frac{2a(0)}{\mu\,\beta} &\text{if}&\gamma=1,\\
&&\\
\displaystyle  \Gamma_c(0)&\text{if}&\gamma>1.
\end{array}
\right.
\end{equation}
Moreover, $\calK_0$ is the unique solution in $\calC^0_z([0,\infty),\calS'_{0,s,\by}(\R\times\R^2)) \cap \calC^1_z((0,\infty),\calS'_{0,s,\by}(\R\times\R^2))$ to the paraxial wave equation
\begin{equation}\label{eq:parax2}
\partial^2_{sz} \calK_0 + \frac{c_0}{2} \Delta_{\by} \calK_0 - \calI_0(\calK_0) = 0,
\end{equation}
with $\calK_0(s,\by,z=0)=\delta(s)\delta(\by)$, and 
\[
\calI_0(\psi):=\left\{ 
\begin{array}{ccc}
\displaystyle \frac{\theta'^2_0 R_0}{8 c_0^2} \partial^3_{sss} \psi & \text{if} & \gamma \geq 1, \\
&&\\
&&\\
\displaystyle  \frac{\theta'^2_0 R_0 \Gamma(1-\gamma)}{2^{3-\gamma}c_0^{1+\gamma}}  D^{2+\gamma}_s \psi  &\text{if} & \gamma \in (0,1).
\end{array}
\right.
\]
\end{theorem}

In this result, we can easily observe the difference between the case $\gamma\geq 1$ and $\gamma \in(0,1)$. In the former case, $\calI_0$ is a classical third order differential operator, while for $\gamma \in (0,1)$, we have a fractional derivative of order $2+\gamma\in (2,3)$. Also, even if the case $\gamma=1$ corresponds to slowly decaying correlations, the kernel $\calK_0$ behaves as for $\gamma>1$, and this case plays somehow  the role of a continuity point w.r.t. the order of derivation in $\calI_0$. Moreover, as we will see in the next section, the random travel time \eqref{def:RTT} has a very large standard deviation w.r.t. the pulse width for $\gamma=1$. Therefore, the case $\gamma=1$ do have the behavior of long-range correlations. 

Another remark, for $\gamma\geq 1$, there is no effective dispersion anymore in the limit $l_0\to 0$, it remains only an effective frequency-dependent attenuation in $\omega^2$. Nevertheless, the effective dispersion is still present for $\gamma\in(0,1)$, and as $\gamma \nearrow 1$, this dispersion remains of order $1$, while the attenuation becomes strong. For long-range correlations, one can observe in \eqref{def:hatK0} the frequency-dependent attenuation given by the power law
\[|\omega|^{1+\gamma}\qquad\gamma\in(0,1],\]
with exponent depending on the decay rate of the correlation function of the medium fluctuations \eqref{eq:decay_R}. 

Unfortunately, our choice of random field $V$ does not allow finer results for short-range correlations as the ones obtained in \cite{garnier2}, with $\lambda\in(0,1)$ in \eqref{eq:powerdecay}. In the case of short-range correlations ($\gamma>1$ in \eqref{eq:decay_R}), we would need the additional requirement 
\[
\int_{-\infty}^\infty R(z) dz = 0,
\]
which cannot be satisfied in our context since $R$ is a positive function.

Finally, due to technical reasons, our approach does not allow to derive the result of Theorem \ref{th:asympt} directly from the wave equation with a proper scaling in $\eps$. Such an approach would require $a(\eps^{1/(2\beta)} p)$ and $S/ \eps^{1/(2\beta)}$ in the definition \eqref{def:V} of $V$. However, in this case, the key technical estimate \eqref{eq:prop1V} would not be valid anymore. This is the reason why the second limit in $l_0$ is introduced.

\section{Travel time analysis}\label{sec:travel_time}

In this section, we discuss the asymptotic behavior of the random travel time \eqref{def:RTT} and its consequences. This analysis has already been carried out in \cite{garnier1, garnier2} under more general fluctuation models. Here, we work out the main lines, under our setting, to provide a complete picture regarding the impact of long-range correlations on the propagating pulse.

As already noticed, the travel time $T^0_\eps(L)$, for the stable wave-front to reach the plan $z=L$, is random. Its precise behavior can be described through the following result, which is proved in Appendix \ref{proof:travel_time}.

\begin{proposition}\label{prop:travel_time}
Let us defined 
\[
W_\eps(L):= \frac{1}{\sigma_\eps}\Big(T_\eps^0(L) - \frac{L}{c_0}\Big),
\]
where
\[
\sigma_\eps := \left\{ \begin{array}{ccc}
\displaystyle \eps^{(1+\gamma)/2} &\text{if}& \gamma\in (0,1),\\
&&\\
\displaystyle \eps|\ln(\seps)|^{1/2} & \text{if} & \gamma=1, \\
&&\\
\displaystyle \eps & \text{if} & \gamma >1. \\
\end{array}\right.
\]
The family $(W_\eps(L))_\eps$ converges in distribution to a limit $W_0(L)$, where:
\begin{itemize}
\item for $\gamma\in(0,1)$, $W_0$ is a fractional Brownian motion with Hurst index
\[H = 1-\frac{\gamma}{2} \in (1/2,1),\]
and
\[
\E[W_0(L)^2] = L^{2H} \frac{\theta'^2_0 R_0}{H(2H-1)},
\]
with $R_0$ defined by \eqref{def:R0};
\item for $\gamma \geq 1$, $W_0$ is a Brownian motion with
\[
\E[W_0(L)^2] = L \, \theta'^2_0 \Gamma_0,
\]
and $\Gamma_0$ defined by \eqref{def:Gamma0}.
\end{itemize} 
\end{proposition}
In other words, the random travel time for the wave-front can be formally expressed as follows for $\gamma>1$,
\[
T^0_\eps(L) = \frac{L}{c_0} + \eps W_0(L) + o(\eps).
\] 
We can observe an effective random time-shift, w.r.t. the expected travel time $L/c_0$, given by a Brownian motion of order the pulse width $\eps$. This observation is consistent with the standard ODA theory. For $\gamma=1$, we now have
\[
T^0_\eps(L) = \frac{L}{c_0} + \eps\, |\ln(\seps)|^{1/2} W_0(L) + o(\eps),
\] 
with still a random time-shift given by a Brownian motion, but with a standard deviation (sd) larger, by a factor $|\ln(\seps)|^{1/2}$, than the pulse width. In other words, we have
\begin{equation}\label{eq:sd_LR1}
\text{sd}\Big[ \frac{T^0_\eps(L)}{\eps} \Big] \propto  |\ln(\seps)|^{1/2} \gg 1.
\end{equation}
This becomes more significant for slowly decaying correlations, with $\gamma\in(0,1)$, since we have
\[
T^0_\eps(L) = \frac{L}{c_0} + \eps \cdot \eps^{-(1-\gamma)/2} W_0(L) + o(\eps^{(1+\gamma)/2}),
\] 
where the random time-shift is now given by a fractional Brownian motion. The standard deviation of this random time-shift is larger than the pulse width by a factor $\eps^{-(1-\gamma)/2}\gg 1$, that is
\begin{equation}\label{eq:sd_LR2}
\text{sd}\Big[ \frac{T^0_\eps(L)}{\eps} \Big] \propto  \eps^{-(1-\gamma)/2} \gg 1.
\end{equation}
To sum up, for short-range correlations we observe a time-shift, w.r.t. the expected travel time $L/c_0$, of order the pulse width. But for long-range correlations, even if this shift remains small compared to the   expected travel time $L/c_0$, it becomes very large compared to the pulse width.

One can also remark that the random travel time \eqref{def:RTT} for the stable wave-front does not correspond exactly to the travel time along random characteristics
\[
T_\eps(L) := \int_0^L \frac{du}{c_\eps(u)} = \frac{1}{c_0} \int_0^L \sqrt{1+\nu_\eps(u/\eps)} du,
\]
representing the arrival time at $z=L$. The random time $T^0_\eps(L)$ provides a convenient approximation to $T_\eps(L)$ for the analysis developed in this paper. Hence, we observe an arrival delay 
\[\Delta T_\eps(L):= T^0_\eps(L) - T_\eps(L)\]
for the wave-front, which can be roughly expressed, after some algebra, as
\[
\Delta T_\eps(L) =  \frac{1}{8 c_0} \int_0^L  \nu^2_\eps(z/\eps) dz + o(\eps).
\]
From this expression, one can see that this arrival delay is positive, for $\eps$ small enough, meaning that compared to the travel time $T_\eps(L)$ the stable wave-front exhibits a delay to reach the plan $z=L$. The comparison of the arrival times w.r.t. the pulse width $\eps$ can be characterized precisely as follows.
\begin{proposition}\label{prop:delay}
We have
\[\lim_{\eps\to 0} \frac{\Delta T_\eps(L)}{\eps} = \frac{\theta'^2_0 R(0) L}{8c_0}\]
in probability.
\end{proposition}
The details of the proof are provided in Appendix \ref{proof:delay}, and this result shows that the wave-front exhibits a deterministic delay of order the pulse width w.r.t. the travel time $T_\eps(L)$.

\section{Modal decomposition} \label{sec:modal_dec}

The stochastic analysis provided in this paper is based on a modal decomposition of the wave field in the space-time frequency domain, which follows the lines of \cite[Chapter 14]{fouque}.    

To study \eqref{eq:waveeq_eps}, we introduce the following specific Fourier transform
\[
\hat f_\eps(\omega,\kappa) := \iint f(t,\bx)e^{i \omega (t/\eps-\kappa \cdot \bx/\seps)} dt d\bx,
\]
and its corresponding inverse formulation
\begin{equation}\label{def:iFT}
f(t,\bx) := \frac{1}{(2\pi)^3\eps^{2}}\iint \hat f_\eps(\omega,\kappa)e^{-i \omega (t/\eps-\kappa \cdot \bx/\seps)} \omega^2 d\omega d\bx,
\end{equation}
which are scaled according to the source term. Applying this Fourier transform to \eqref{eq:waveeq_eps} gives
\begin{equation}\label{eq:Fwaveeq_eps}
\partial^2_{zz} \hat p_\eps + \frac{\omega^2\lambda^2_\eps(\kappa)}{\eps^2}\hat p_\eps + \frac{\omega^2}{\eps^2 c^2_0}\nu_\eps\Big(\frac{z}{\eps}\Big)\mathbf{1}_{(0, L)}(z) \hat p_\eps  = \eps^{2} \hat \Psi(\omega,\kappa) \delta(z) \qquad (t,\bx,z)\in \R\times \R ^2 \times \R.
\end{equation}
with 
\[
\hat \Psi(\omega,\kappa) :=\iint \Psi(t,\bx)e^{i \omega (t-\kappa \cdot \bx)} dt d\bx,
\]
the unscaled Fourier transform of the source profile $\Psi$, and 
\begin{equation}\label{def:lambda_eps_k}
\lambda_\eps(\kappa) := \frac{\sqrt{1 - \eps c^2_0 |\kappa|^2 }}{c_0}.
\end{equation}
Throughout this paper, for simplicity, we assume that $\hat \Psi$ is compactly supported within a ball centered at $0$ and radius of order $1$, that is not depending on $\eps$. We also assume for technical reasons that $\omega=0$ does not belong to the support of the source:
\[
supp_\omega \hat \Psi \subset (-\infty ,-\omega_c)\cup (\omega_c,\infty),
\]
for some cutoff frequency $\omega_c>0$. This assumption allows to avoid unnecessary complications to define \eqref{eq:schrodinger} and in the proof of Theorem \ref{th:asympt}. Therefore, we have
\[
\hat p_\eps(\omega,\kappa,z) = 0\qquad\text{for}\qquad |\kappa| \geq \frac{1}{c_0 \, \seps}.
\]
These assumptions are not restrictive and do not change the overall result, but simplify greatly the presentation. In fact, for $|\kappa| < 1/(c_0 \, \seps)$, we only deal with the oscillatory components of the solution to \eqref{eq:Fwaveeq_eps}. The components associated to $|\kappa| > 1/(c_0 \, \seps)$ correspond to the evanescent modes that decay exponentially w.r.t the $z$-variable. Due to this exponential decay, the evanescent modes do not contribute in a significant way in the limit $\eps\to 0$, and are therefore considered as negligible.        

Note also that the source term in \eqref{eq:Fwaveeq_eps} produces the following jump conditions at the source location $z=0$, that are used below to determine the initial amplitudes of the modal decomposition:
\begin{equation}\label{eq:jump_cond}
\begin{split}
\hat p_\eps(\omega,\kappa,z=0^+)- \hat p_\eps(\omega,\kappa,z=0^-) & = 0,  \\
\partial_z\hat p_\eps(\omega,\kappa,z=0^+)- \partial_z\hat p_\eps(\omega,\kappa,z=0^-) & = \eps^2 \hat \Psi(\omega,\kappa).
\end{split}
\end{equation}

\subsection{Mode coupling in random media}

In the random section, that is for $z \in (0,L)$, we decompose the solutions to the second order equation \eqref{eq:Fwaveeq_eps} as right- and left-going modes,
\begin{equation}\label{def:dec_p}
\hat p_\eps(\omega,\kappa,z):= \frac{1}{\sqrt{\omega \lambda_\eps(\kappa)}}\Big( \hat a_\eps(\omega,\kappa, z) e^{i\omega \lambda_\eps(\kappa) z/ \eps} + \hat b_\eps(\omega,\kappa, z) e^{-i\omega \lambda_\eps(\kappa) z/\eps}\Big),
\end{equation}
with the additional condition
\[
\frac{d}{dz} \hat a_\eps(\omega,\kappa, z) e^{i\omega \lambda_\eps(\kappa) z/ \eps} + \frac{d}{dz} \hat b_\eps(\omega,\kappa, z) e^{-i\omega \lambda_\eps(\kappa) z/\eps} =0,
\]
so that 
\[
\partial_z \hat p_\eps(\omega,\kappa,z)= \frac{i\sqrt{\omega \lambda_\eps(\kappa)} }{\eps}\Big(\hat a_\eps(\omega,\kappa, z) e^{i\omega \lambda_\eps(\kappa) z/ \eps} - \hat b_\eps(\omega,\kappa, z) e^{-i\omega \lambda_\eps(\kappa) z/\eps}\Big).
\]
Hence, both $\hat a_\eps$ and $\hat b_\eps$ can be expressed in terms of $\hat p_\eps$ and $\partial_z \hat p_\eps$:
\begin{align*}
\hat a_\eps(\omega,\kappa, z) &= \frac{1}{2}\Big( \sqrt{\omega \lambda_\eps(\kappa)} \hat p_\eps(\omega,\kappa,z) + \frac{\eps}{i\sqrt{\omega \lambda_\eps(\kappa)}}\partial_z\hat p_\eps(\omega,\kappa,z)\Big)e^{-i\omega \lambda_\eps(\kappa) z/ \eps},\\
\hat b_\eps(\omega,\kappa, z) & = \frac{1}{2}\Big( \sqrt{\omega \lambda_\eps(\kappa)}\hat p_\eps(\omega,\kappa,z) - \frac{\eps}{i\sqrt{\omega \lambda_\eps(\kappa)}}\partial_z\hat p_\eps(\omega,\kappa,z)\Big)e^{i\omega \lambda_\eps(\kappa) z/ \eps}.
\end{align*}

Here, $\hat a_\eps$ represents the amplitudes of the right-going modes, while $\hat b_\eps$ the ones of the left-going modes. Differentiating in $z$, these two last expressions, and using \eqref{eq:Fwaveeq_eps} give

\begin{equation}\label{eq:coupledmode}
\frac{d}{dz}\begin{pmatrix}
\hat a_\eps(\omega,\kappa, z) \\
\hat b_\eps(\omega,\kappa, z)
\end{pmatrix}
= \frac{1}{\eps}\nu_\eps\Big(\frac{z}{\eps}\Big) H_\eps\Big(\omega,\kappa,\frac{z}{\eps}\Big)
\begin{pmatrix}
\hat a_\eps(\omega,\kappa, z) \\
\hat b_\eps(\omega,\kappa, z)
\end{pmatrix},
\end{equation}
where
\begin{equation}
\label{def:Hmat}
H_\eps(\omega,\kappa,z) = \frac{i\omega}{2\lambda_\eps(\kappa)c_0^2} \begin{pmatrix}
1 & e^{-2i\omega \lambda_\eps(\kappa)z} \\
- e^{2i\omega \lambda_\eps(\kappa)z} & -1
\end{pmatrix}.
\end{equation} 
This differential equation describes how a wave is affected while going through the slab $(0,L)$. More precisely, it describes how the medium fluctuations produce the scattering effects on the propagating wave through the exchange between the right- and left-going modes. Note that there is no coupling between any two distinct $\kappa$-modes since we consider a randomly layered propagation medium. Moreover, the scattering slab is surrounded by two homogeneous half-spaces, so that we need to complement this system with boundary conditions representing the incoming waves in the slab and the outgoing waves from the slab.

\subsection{Boundary conditions}

In this section, we depict the propagation mechanism in absence of random fluctuations of the wave-speed profile, which corresponds to the situation for $z<0$ and $z>L$.  
\begin{figure}
\begin{center}
\includegraphics[scale=0.25]{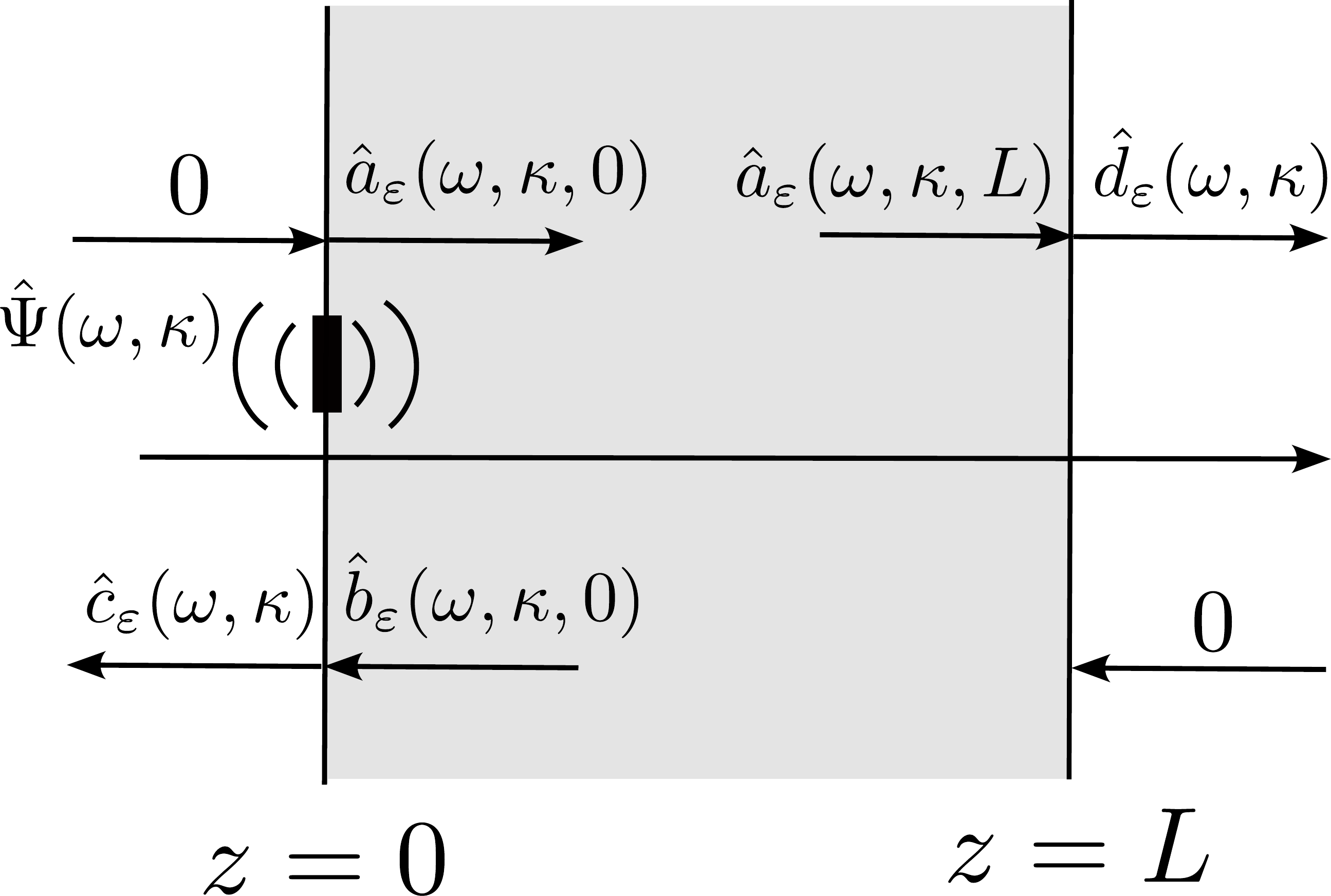}
\end{center}
\caption{\label{fig2} Illustration of the mode amplitudes associated to the incoming and outgoing waves at $z=0$ and $z=L$.}
\end{figure}
Considering a decomposition similar to \eqref{def:dec_p}, in absence of random fluctuations, leads to constant mode amplitudes in view of \eqref{eq:coupledmode}. Therefore, if we assume that no waves are coming from the left- and the right-hand side of the random slab $z\in(0,L)$ (see Figure \ref{fig2} for an illustration), we necessarily have
\[
\hat p_\eps(\omega,\kappa,z):= \frac{\hat c_\eps(\omega,\kappa)}{\sqrt{\omega \lambda_\eps(\kappa)}} e^{-i\omega \lambda_\eps(\kappa) z/ \eps}\qquad\text{for} \quad z<0,
\]
and
\[
\hat p_\eps(\omega,\kappa,z):= \frac{\hat d_\eps(\omega,\kappa)}{\sqrt{\omega \lambda_\eps(\kappa)}} e^{i\omega \lambda_\eps(\kappa) z/ \eps}\qquad\text{for} \quad z>L.
\]
As we will see below, on the left-hand-side of the source term $(z<0)$, we only have a left-going mode produced by the source and the backscattered field at $z=0$ with amplitude $\hat b_\eps(\omega,\kappa,0)$, no right-going mode coming from the left. In the same way, on the right-hand-side of the random section ($z>L$), we only have a right-going mode produced by the wave outgoing the random section at $z=L$, and no left-going mode coming from the right. To be more precise, reminding the expression of $\hat p_\eps$ at each side of the interface $z=L$, 
\begin{align*}
\hat p_\eps(\omega,\kappa,z) &= \frac{1}{\sqrt{\omega \lambda_\eps(\kappa)}}\Big( \hat a_\eps(\omega,\kappa, z) e^{i\omega \lambda_\eps(\kappa) z/ \eps} + \hat b_\eps(\omega,\kappa, z) e^{-i\omega \lambda_\eps(\kappa) z/\eps}\Big)\1_{(0,L)}(z)\\
&+\frac{\hat d_\eps(\omega,\kappa)}{\sqrt{\omega \lambda_\eps(\kappa)}} e^{i\omega \lambda_\eps(\kappa) z/ \eps}\1_{(L,\infty)}(z),
\end{align*}
and by continuity of $\hat p_\eps$ and $\partial_z \hat p_\eps$ at $z=L$, we have
\[
\hat d_\eps(\omega,\kappa) := \hat a_\eps(\omega,\kappa, z=L)\qquad \text{and} \qquad \hat b_\eps(\omega,\kappa, z=L)=0.
\]
To see how the source term charges the modes, we remind the expression of $\hat p_\eps$ at each side of the source position $z=0$,
\begin{align*}
\hat p_\eps(\omega,\kappa,z)&= \frac{\hat c_\eps(\omega,\kappa)}{\sqrt{\omega \lambda_\eps(\kappa)}} e^{-i\omega \lambda_\eps(\kappa) z/ \eps}\1_{(-\infty,0)}(z)\\ 
&+\frac{1}{\sqrt{\omega \lambda_\eps(\kappa)}}\Big( \hat a_\eps(\omega,\kappa, z) e^{i\omega \lambda_\eps(\kappa) z/ \eps} + \hat b_\eps(\omega,\kappa, z) e^{-i\omega \lambda_\eps(\kappa) z/\eps}\Big)\1_{(0,L)}(z),
\end{align*}
and use the jump conditions \eqref{eq:jump_cond} to obtain
\[
\hat a_\eps(\omega,\kappa, z=0) = \frac{\eps^{2}\sqrt{\omega \lambda_\eps(\kappa)}}{2}\hat \Psi(\omega,\kappa),
\]
and
\[
\hat c_\eps(\omega,\kappa) = \hat b_\eps(\omega,\kappa, z=0) + \frac{\eps^{2}\sqrt{\omega \lambda_\eps(\kappa)}}{2}\hat \Psi(\omega,\kappa).
\]

\subsection{Wave propagation in homogeneous media}

In this section, we provide a derivation of the paraxial approximation in the case of a homogeneous propagation medium ($\nu_\eps\equiv 0$). In this context, $\hat b_\eps \equiv 0$ from \eqref{eq:coupledmode}, and we simply have  
\[
\hat p_\eps(\omega,\kappa,z) = \frac{\eps^{2}}{2}\hat \Psi(\omega,\kappa)e^{i\omega \lambda_\eps(\kappa) z/ \eps} \qquad\text{for}\quad z>0,
\]
so that taking the inverse Fourier transform \eqref{def:iFT} gives
\[
p_\eps(t,\bx,z) = \frac{1}{2(2\pi)^3} \int \hat \Psi(\omega,\kappa)e^{i\omega \lambda_\eps(\kappa)z/\eps}e^{-i\omega (t/\eps - \kappa\cdot \bx/\seps)}\omega^2 d\omega d\kappa.
\]
Let us remark that
\begin{equation}\label{eq:lambda_asymp}
\lambda_\eps(\kappa) = \frac{1}{c_0} - \frac{\eps\, c_0}{2}  |\kappa|^2 + \calO(\eps^{2}),
\end{equation}
thanks to \eqref{def:lambda_eps_k}. Looking at the wave in the frame of the source term, by setting
\[
t = \frac{z}{c_0} + \eps \,s \qquad\text{and}\qquad \bx = \seps \,\by,
\]
we obtain 
\begin{align*}
\lim_{\eps\to 0} p_\eps\Big(\frac{z}{c_0} + \eps s, \seps \by, z\Big) &= \psi(s,\by,z)\\
&:= \frac{1}{2(2\pi)^3} \int \hat \Psi(\omega,\kappa) e^{-i\omega s} e^{-i\omega ( c_0 |\kappa|^2 z /2 - \kappa\cdot \by)}\omega^2 d\omega d\kappa.
\end{align*}
Here, $\psi$ satisfies the paraxial wave equation
\[\partial^2_{s z} \psi(s,\by,z) +\frac{c_0}{2}\Delta_{\by}\psi(s,\by,z) =0, \qquad\text{with}\qquad \psi(s,\by,z=0)=\frac{1}{2}\Psi(s,\by).\]
In other words, the pulse front can be described through the so-called paraxial wave equation, with a condition at $z=0$ given by half the source profile. Note that if we take the Fourier transform of $\psi$ w.r.t. time (the $s$-variable) we obtain the following Schrödinger equation
\[i\frac{\omega}{c_0} \partial_{z} \check\psi(\omega,\by,z) + \frac{1}{2}\Delta_{\by}\check \psi(\omega,\by,z) =0, \qquad\text{with}\qquad \check \psi(\omega,\by,z=0)=\frac{1}{2}\check\Psi(\omega,\by),\]
where
\[\check \psi(\omega,\by,z):= \int \psi(s,\by,z) e^{i\omega s}ds.\]

\section{Propagator matrix} \label{sec:propagator}

The system \eqref{eq:coupledmode} is a boundary value problem, with 
\begin{equation}\label{eq:bound_cond_ab}
\hat a_\eps(\omega,\kappa, z=0) = \frac{\eps^{2}\sqrt{\omega \lambda_\eps(\kappa)}}{2}\hat \Psi(\omega,\kappa)\qquad\text{and}\qquad\hat b_\eps(\omega,\kappa, z=L)=0, 
\end{equation}
which is not convenient for our analysis based on martingale techniques and diffusion processes corresponding to initial value problems. In this section, we introduce initial value problems that can be related to \eqref{eq:coupledmode}. First, we introduce the associated propagator matrix $\bP_\eps$, which is the solution to 
\[
\frac{d}{dz} \bP_\eps(\omega,\kappa, z)
= \frac{1}{\eps}\nu_\eps\Big(\frac{z}{\eps}\Big) H_\eps\Big(\omega,\kappa,\frac{z}{\eps}\Big)
\bP_\eps(\omega,\kappa, z),\qquad\text{with}\qquad \bP_\eps(\omega,\kappa, z=0)=\mathbf{I}_2,
\]
where $\mathbf{I}_2$ is the $2\times 2$ identity matrix. The relation between the left- and right-going modes with the propagator is given by
\begin{equation}\label{eq:prop_mat_ab}
\begin{pmatrix}
\hat a_\eps(\omega,\kappa, z) \\
\hat b_\eps(\omega,\kappa, z)
\end{pmatrix} = \bP_\eps(\omega,\kappa, z)\begin{pmatrix}
\hat a_\eps(\omega,\kappa, 0) \\
\hat b_\eps(\omega,\kappa, 0)
\end{pmatrix}. 
\end{equation}
From the symmetries of $H_\eps$, given by \eqref{def:Hmat}, the propagator matrix can be recast as 
\[
\bP_\eps(\omega,\kappa, z) = \begin{pmatrix}
\alpha_\eps(\omega,\kappa, z) & \overline{\beta_\eps(\omega,\kappa, z)}\\
\beta_\eps(\omega,\kappa, z) & \overline{\alpha_\eps(\omega,\kappa, z)}
\end{pmatrix},
\] 
where $(\alpha_\eps, \beta_\eps)$ being the solution to
\begin{equation}\label{eq:coupledmode2}
\frac{d}{dz}\begin{pmatrix}
\alpha_\eps(\omega,\kappa, z)\\
\beta_\eps(\omega,\kappa, z)
\end{pmatrix}
= \frac{1}{\eps}\nu_\eps\Big(\frac{z}{\eps}\Big) H_\eps\Big(\omega,\kappa,\frac{z}{\eps}\Big)
\begin{pmatrix}
\alpha_\eps(\omega,\kappa, z) \\
\beta_\eps(\omega,\kappa, z)
\end{pmatrix},
\end{equation}
with
\[
\begin{pmatrix}
\alpha_\eps(\omega,\kappa, 0)\\
\beta_\eps(\omega,\kappa, 0)
\end{pmatrix} = \begin{pmatrix}
1\\
0
\end{pmatrix}.
\]
From this equation, using that $H_\eps$ has null trace, the determinant of the propagator is then constant in $z$,
\[\det \bP_\eps(\omega,\kappa, z) = \det \bP_\eps(\omega,\kappa, 0)= 1, \]
yielding the conservation relation 
\begin{equation}\label{eq:cons1}
|\alpha_\eps(\omega,\kappa, z)|^2 - |\beta_\eps(\omega,\kappa, z)|^2 = 1.
\end{equation}
From these new variables, $\alpha_\eps$ and $\beta_\eps$, one can describe the transmitted mode amplitudes at $z=L$, and the reflected mode amplitudes at $z=0$ using \eqref{eq:bound_cond_ab} and \eqref{eq:prop_mat_ab}:
\[
\hat a_\eps(\omega,\kappa, z=L) = \frac{1}{\overline{\alpha_\eps(\omega,\kappa, L)}}\hat a_\eps(\omega,\kappa, z=0),
\]
and
\[
\hat b_\eps(\omega,\kappa, 0) = -\frac{\beta_\eps(\omega,\kappa, L)}{\overline{\alpha_\eps(\omega,\kappa, L)}}\hat a_\eps(\omega,\kappa, z=0).
\]
One can also remark from \eqref{eq:cons1}, that we have the following conservation relation for the right- and left-going modes
\[
|\hat a_\eps(\omega,\kappa, L)|^2 + |\hat b_\eps(\omega,\kappa, 0)|^2 = |\hat a_\eps(\omega,\kappa, 0)|^2,
\]
telling us that the input energy at $z=0$ equals the sum of the transmitted energy at $z=L$ and the reflected one at $z=0$. 

As already discussed in Section \ref{sec:travel_time}, the random travel time of the wave-front produces terms that can blow up in our scaling regime. To overcome this difficulty in our analysis, we reformulate \eqref{eq:coupledmode2} through the variables
\begin{equation}\label{def:AB}
A_\eps(\omega,\kappa,z) := \alpha_\eps(\omega,\kappa, z)e^{-i\omega \phi_\eps(\kappa,z)/\eps}\quad\text{and}\quad B_\eps(\omega,\kappa,z) := \beta_\eps(\omega,\kappa, z)e^{i\omega \phi_\eps(\kappa,z)/\eps},
\end{equation}
where
\begin{equation}\label{def:phi}
\phi_\eps(\kappa,z) := \frac{1}{2\lambda_\eps(\kappa)c_0^2} \int_0^z \nu_\eps(s/\eps^2)ds.
\end{equation}
While this quantity provides an effective limit as $\eps\to 0$ for mixing fluctuations, it blows up in case of long-range correlations, and this is the reason why we single out its contribution from \eqref{eq:coupledmode2}. In this latter context, the term \eqref{def:phi} is responsible of the large standard deviation of the random travel time $T^0_\eps$ w.r.t. to the pulse width (\ref{eq:sd_LR1}, \ref{eq:sd_LR2}).

The new variables $(A_\eps, B_\eps)$ satisfy the system 
\begin{equation}\label{eq:coupledmodefilt}
\frac{d}{dz}\begin{pmatrix}
A_\eps(\omega,\kappa, z) \\
B_\eps(\omega,\kappa, z)
\end{pmatrix}
= \frac{1}{\eps}\nu_\eps\Big(\frac{z}{\eps}\Big) \calH_\eps \Big(\omega,\kappa,\frac{\tau_\eps(\kappa,z)}{\eps}\Big)
\begin{pmatrix}
A_\eps(\omega,\kappa, z) \\
B_\eps(\omega,\kappa, z)
\end{pmatrix},
\end{equation}
with
\[
\calH_\eps(\omega,\kappa,\tau) = \frac{i\omega}{2\lambda_\eps(\kappa)c_0^2} \begin{pmatrix}
0 & e^{-i\omega \tau} \\
- e^{i\omega \tau} & 0
\end{pmatrix},
\]
and 
\begin{equation}\label{def:tau}
\tau_\eps(\kappa,z):= 2\lambda_\eps(\kappa)z + \phi_\eps(\kappa,z).
\end{equation}
Note that from \eqref{eq:cons1} and \eqref{def:AB}, we still have the conservation relation
\begin{equation}\label{eq:conservation}
|A_\eps(\omega,\kappa, z)|^2 - |B_\eps(\omega,\kappa, z)|^2 = 1.
\end{equation}

Let us remark that it is not clear how the strategy proposed by \cite{marty0}, based on the rough-path theory, could be applied to the system \eqref{eq:coupledmodefilt} in case of long-range correlations. Compared to \cite{marty0}, we have the additional random blowing term $\phi_\eps(\kappa,z)$ in the periodic component of $\calH_\eps$, which makes the coupling matrix in \eqref{eq:coupledmodefilt} nonlinear in $\nu_\eps$. This nonlinear behavior and the long-range correlation property make difficult the evaluation of key quantities allowing the use of the Terry-Lyons continuity theorem (see \cite{marty0}). This is the reason why we do not follow this route in this paper. 

Finally, to study the asymptotic behavior of the mode amplitudes $\hat a_\eps(\omega,\kappa,L)$ and $\hat b_\eps(\omega,\kappa,0)$, one can study the one of $(A_\eps,B_\eps)$ as $\eps\to 0$, which is given by the following result.
\begin{theorem}\label{th:asymp}
Let $n\geq 1$, and set
\[
X_\eps(\omega,\kappa, z)=\begin{pmatrix}
A_\eps(\omega,\kappa, z) \\
B_\eps(\omega,\kappa, z)
\end{pmatrix}.
\]
For any $(\omega_1,\dots,\omega_n)$ and $(\kappa_1,\dots,\kappa_n)$, the process $\calX_\eps$, defined by 
\[
\calX_\eps(z) := \big(X_\eps(\omega_1,\kappa_1, z),\dots,X_\eps(\omega_n,\kappa_n, z)\big), 
\]
converges in distribution in $\calC((0,\infty),\mathbf{C}^{2n})$ to a process 
\begin{equation}\label{def:calX0}
\calX_0(z) := \big(X_0(\omega_1, z),\dots,X_0(\omega_n, z)\big), 
\end{equation}
independent of the $\kappa$-variables, and where all its components are statistically independent. Here, for each $\omega$, $X_0(\omega,\cdot)$ is solution to the following stochastic differential equation
\begin{equation}\label{eq:EDS}
\begin{split}
d X_0(\omega, z) &= -\sqrt{\frac{\theta'^2_0 \omega^2 \Gamma_c(\omega)}{ 4c_0^2}}\begin{pmatrix} 0 & 1 \\ 1 & 0 \end{pmatrix}X_0(\omega, z) \circ dW_1(z) \\
&- i \sqrt{\frac{\theta'^2_0\omega^2 \Gamma_c(\omega)}{ 4c_0^2}}\begin{pmatrix} 0 & 1 \\ -1 & 0 \end{pmatrix}X_0(\omega, z) \circ dW_2(z) \\
& - i\frac{\theta'^2_0 \omega^2 \Gamma_s(\omega)}{8c_0^2}\begin{pmatrix} 1 & 0 \\ 0 & 1 \end{pmatrix}X_0(\omega, z) dz,
\end{split}
\end{equation}
where $W_1$ and $W_2$ are two independent real-valued standard Brownian motions, $\circ$ stands for the Stratonovich integral,
\[
\Gamma_c(\omega):=2\int_0^\infty R(s)\cos\Big(\frac{2 \omega s}{c_0}\Big) ds,\qquad\text{and}\qquad \Gamma_s(\omega):=2\int_0^\infty R(s) \sin\Big(\frac{2 \omega s}{c_0}\Big) ds,
\]
with $R$ given by \eqref{def:R}.
\end{theorem}

This result is consistent with the one of \cite[Section 7.1]{fouque}. The only difference with \cite[Equation 7.16]{fouque} comes from the extra term 
\[
i\sqrt{\frac{\theta'^2_0\omega^2 \Gamma_c(0)}{ 4c_0^2}}\begin{pmatrix} 1 & 0 \\ 0 & -1 \end{pmatrix}X_0(\omega, z) \circ dW_0(z),
\]
which is not in \eqref{eq:EDS}, and where $W_0$ is a real-valued standard Brownian motion independent of $(W_1,W_2)$. This extra term is responsible for the Brownian arrival time-shift in the standard ODA theory, as exhibited through the random travel time of the wave-front in Section \ref{sec:travel_time}. In Theorem \ref{th:asymp}, even under short-range correlations ($\gamma>1$ in \eqref{eq:decay_R}), this term has disappeared since we have considered the compensated mode amplitudes \eqref{def:AB}. Note also that in \cite[Section 7.1]{fouque} the Brownian time-shift is the same for each frequencies $\omega_j$, which then correlates all the mode amplitudes at different frequencies. This is in contrast with our result where the contribution of the random travel time has been compensated. In our context, the statistical independence for any distinct frequencies is responsible for the pulse stabilization, that is the convergence in probability toward a deterministic limit in Theorem \ref{th:main}. In the context of long-range correlations ($\gamma \in (0,1]$), the compensations \eqref{def:AB} are mandatory due to the fact that $\Gamma_c(0)=\infty$ in this case. 

Before going into the proof of Theorem \ref{th:asymp}, which is given in Section \ref{proof:th_asymp}, we show how this result plays a role in the proof of Theorem \ref{th:main}.

\section{The transmitted waves and proof of Theorem \ref{th:main}}\label{proof:th_main}

The transmitted wave, on the time frame of the random travel time $ T^0_\eps(L) $ and scaled according to the source profile \eqref{eq:waveeq_eps}, is given by
\begin{align*}
p^L_{tr,\eps}(s,\by)& := p_\eps\big( T^0_\eps(L) + \eps s, \seps \by, L \big)\\
&= \frac{1}{(2\pi)^3 \eps^2} \int \frac{\hat a_\eps(\omega,\kappa,L)}{\sqrt{\omega \lambda_\eps(\kappa)}} e^{i\omega (\lambda_\eps(\kappa) L - T^0_\eps(L))/\eps }e^{-i\omega ( s - \kappa \cdot\by)} \omega^2 d\omega \,d\kappa \\
& =  \frac{1}{2(2\pi)^3} \int \frac{\hat \Psi(\omega,\kappa)}{\overline{A_\eps(\omega,\kappa,L)}} e^{i\omega \Phi_\eps(\kappa, L)/\eps }e^{-i\omega ( s - \kappa \cdot\by)} \omega^2 d\omega \,d\kappa,
\end{align*}
where
\[
\Phi_\eps(\kappa, L) := \lambda_\eps(\kappa) L - T^0_\eps(L) + \phi_\eps(\kappa,L),
\]
with $T^0_\eps(L)$ defined by \eqref{def:RTT}, and $\phi_\eps(\kappa,L)$ defined by \eqref{def:phi}. For $\Phi_\eps(\kappa, L)$ we have the following result.
\begin{lemma}\label{lem:Phi}
We have
\[
\lim_{\eps\to 0}\frac{1}{\eps}\E\Big[ \Big|\Phi_\eps(\kappa,L) + \frac{L c_0}{2} |\kappa|^2 \Big| \Big] =0.
\]
\end{lemma}
\begin{proof}
By definition we have
\[
\Phi_\eps(\kappa, L)=L\Big(\lambda_\eps(\kappa) - \frac{1}{c_0}\Big) + \frac{1}{2c_0}\Big( \frac{1}{c_0\lambda_\eps(\kappa)} - 1 \Big)\int_0^{L} \nu_\eps\Big(\frac{s}{\eps}\Big)ds,
\]
with from \eqref{def:lambda_eps_k}
\[
\lambda_\eps(\kappa) = \frac{1}{c_0} - \eps \frac{c_0 |\kappa|^2 }{2} + \calO(\eps^2),
\]
where $\calO$ is uniform in $\kappa$ since $\hat \Psi$ is compactly supported in both $\omega$ and $\kappa$ in a ball with radius of order 1 w.r.t. $\eps$. As a result,
\[\frac{1}{\eps}\E\Big[ \Big|\Phi_\eps(\kappa, L) + \frac{L c_0}{2} |\kappa|^2 \Big| \Big] = \seps \sup_{s\in[0,L]}\E[|V(s/\eps)|] \sup |\Theta'| + \calO(\eps),\]
which concludes the proof of the lemma according to \eqref{eq:prop1V}.
\end{proof} 
The conservation relation \eqref{eq:conservation} implies that $1/\overline{A_\eps}$ is uniformly bounded by $1$ in all its variables, so that thanks to Lemma \ref{lem:Phi},
\[
\lim_{\eps\to 0}\E\Big[\sup_{s,\by}|p^L_{tr,\eps}(s,\by) - q^L_{tr,\eps}(s,\by)| \Big] = 0,
\] 
with
\[
q^L_{tr,\eps}(s,\by) := \frac{1}{2(2\pi)^3} \int \frac{\hat \Psi(\omega,\kappa)  }{\overline{A_\eps(\omega,\kappa,L)}}  e^{-i\omega ( s + L c_0|\kappa|^2/2 -\kappa \cdot\by)} \omega^2 d\omega \,d\kappa. 
\]
Hence, it is enough to prove the convergence for $q^L_{tr,\eps}$ to obtain the one of $p^L_{tr,\eps}$ according to \cite[Theorem 3.1 pp. 27]{billingsley}.

\begin{proposition}\label{prop:asymp}
The family $(q^L_{tr,\eps})_\eps$ converges in probability in $\calC(\R\times \R^2)$ to
\[
p^L_{tr}(s,\by):=\frac{1}{2(2\pi)^3} \int \hat \Psi(\omega,\kappa)  e^{-\theta'^2_0\omega^2( \Gamma_c(\omega) + i\Gamma_s(\omega)) L/(8c_0^2)}  e^{-i\omega ( s + L c_0|\kappa|^2/2 -\kappa \cdot\by)} \omega^2 d\omega \,d\kappa. 
\]
\end{proposition}
\begin{proof}
Denoting
\[
\calE(\omega) =  e^{-\theta'^2_0\omega^2( \Gamma_c(\omega) + i\Gamma_s(\omega)) L/(8c_0^2)},
\]
and
\begin{align*}
\bE_\eps(\omega_1,\omega_2,\kappa_1,\kappa_2) = \E\Big[ \Big(  \frac{1}{\overline{A_\eps(\omega_1,\kappa_1,L)}} - \calE(\omega_1)  \Big)\Big(  \frac{1}{A_\eps(\omega_2,\kappa_2,L)} - \overline{\calE(\omega_2)}  \Big) \Big],
\end{align*}
we have
\begin{align*}
\E\Big[\sup_{s,\by}|q^L_{tr,\eps}(s,\by)-p^L_{tr}(s,\by)|^2 \Big] \leq \frac{1}{4(2\pi)^6} & \int  | \bE_\eps(\omega_1,\omega_2,\kappa_1,\kappa_2) | \\
&  \times |\hat \Psi(\omega_1,\kappa_1) \overline{ \hat \Psi(\omega_2,\kappa_2)}| \omega^2_1 \omega^2_2 d\omega_1 \,d\kappa_1 \, d\omega_2 \,d\kappa_2.
\end{align*}
Now, we expand the expectation in $\bE_\eps$ so that
\begin{align*}
\bE_\eps(\omega_1,\omega_2,\kappa_1,\kappa_2) & = \E\Big[\frac{1}{\overline{A_\eps(\omega_1,\kappa_1,L)}A_\eps(\omega_2,\kappa_2,L)}\Big] - \E\Big[\frac{1}{\overline{A_\eps(\omega_1,\kappa_1,L)}}\Big]\overline{\calE(\omega_2)} \\
&- \calE(\omega_1) \E\Big[\frac{1}{A_\eps(\omega_2,\kappa_2,L)}\Big]
+ \calE(\omega_1)\overline{\calE(\omega_2)}.
\end{align*}
From Theorem \ref{th:asymp}, with a single $\omega$ and $\kappa$, together with the It\^o formula \cite[Theorem 3.3 pp. 149]{karatzas} applied to the real and imaginary parts of \eqref{eq:EDS}, we obtain
\[
\lim_{\eps\to 0 }\E\Big[\frac{1}{\overline{A_\eps(\omega,\kappa,L)}}\Big]=\calE(\omega).
\]
Note that the convergence in distribution implies the convergence of the expectation thanks to \cite[Theorem 3.5 pp. 31]{billingsley} and the fact that $1/\overline{A_\eps}$ is uniformly bounded in $\eps$ according to \eqref{eq:conservation}. For similar reasons, but adding the fact that the limit in $\eps$ of $(A_\eps(\omega_1,\kappa_1,L),A_\eps(\omega_2,\kappa_2,L))$ for two distinct pairs of $(\omega,\kappa)$ are independent (as stated in Theorem \ref{th:asymp}), we have
\begin{align*}
\lim_{\eps\to 0}\E\Big[\frac{1}{\overline{A_\eps(\omega_1,\kappa_1,L)}A_\eps(\omega_2,\kappa_2,L)}\Big] &= \lim_{\eps\to 0}\E\Big[\frac{1}{\overline{A_\eps(\omega_1,\kappa_1,L)}}\Big]\lim_{\eps\to 0}\E\Big[\frac{1}{A_\eps(\omega_2,\kappa_2,L)}\Big]\\
& = \calE(\omega_1)\overline{\calE(\omega_2)}.
\end{align*}
Combining these results gives
\[
\lim_{\eps\to 0}\bE_\eps(\omega_1,\omega_2,\kappa_1,\kappa_2) = 0,
\]
and thanks to the dominated convergence theorem, 
\[
\lim_{\eps\to 0}\E\Big[\sup_{s,\by}|q^L_{tr,\eps}(s,\by)-p^L_{tr}(s,\by)|^2 \Big] = 0,
\]
which concludes the proof of the proposition. 
\end{proof} 

The transmitted wave-front at the end of the random slab $z=L$ is then given by
\[
p^L_{tr}(s,\by) = \frac{1}{2}\calK(\cdot, \cdot, L) \ast \Psi(s,\by),
\]
with the pulse deformation given in the Fourier domain by
\[
\hat \calK(\omega, \kappa, z) = e^{-\theta'^2_0 \omega^2 ( \Gamma_c(\omega)+i  \Gamma_s(\omega)) z/(8 c_0^2)}  e^{-i\omega c_0|\kappa|^2 z/2}.
\]
From this formula, it turns out that $\calK$ is the unique solution to
\[
\partial^2_{s z} \calK + \frac{c_0}{2} \Delta_{\by} \calK - \calI (\calK) = 0,
\] 
with
\begin{equation}\label{eq:defI1}
\calI (\psi)(s,\by)  := -\frac{\theta'^2_0 }{8c_0^3} \int  i\omega^3( \Gamma_c(\omega)  +i \Gamma_s(\omega) ) \check \psi(\omega,\by) e^{i\omega s} d\omega,
\end{equation}
and $\calK(s,\by, z=0)=\delta(s)\delta(\by)$. The details for the uniqueness are provided in Appendix \ref{sec:uniqueness}. In the time domain, the operator $\calI$ can be recast as 
\begin{equation}\label{eq:defI2}
\calI (\psi)(s) = \phi \ast \partial^3_{sss} \psi(s) \qquad\text{with}\qquad \phi(s)= \frac{\theta'^2_0}{8c_0^2} R\Big(\frac{c_0 s}{ 2}\Big) \, \1_{(0,\infty)}(s),
\end{equation}
so that 
\[
\calI (\psi)(s) = \frac{\theta'^2_0}{8c_0^2} \int_{-\infty}^s R\Big(\frac{c_0 (s-\tau)}{2}\Big) \partial^3_{sss} \psi(\tau) d\tau.
\] 

Let us finish this section with a comment on the backscattered wave
\[
p_{bk,\eps}(s,\by) := p_\eps\big(\eps \,s, \seps \,\by, 0 \big).
\] 
According to Theorem \ref{th:asymp}, the amplitude $B_\eps/\overline{A_\eps}$ converges in distribution to some limit which has null expectation according to the It\^o formula. Note that we again have the convergence of the expectation, since $B_\eps/\overline{A_\eps}$ is also uniformly bounded in $\eps$ by $1$ according to \eqref{eq:conservation}. Thanks to the independence of the limit w.r.t. to multiple frequencies $\omega$, the limiting expectation of the product of two $B_\eps/\overline{A_\eps}$ at two distinct frequencies is the product of the limiting single expectation, which is then $0$. Following the same lines as the proof of Proposition \ref{prop:asymp}, we obtain that the backscattered wave converges to $0$, as $\eps\to 0$, in probability in $\calC(\R\times\R^2)$. This observation is consistent with \cite[Chapter 9]{fouque} in which the authors describe the backscattered signal as a \emph{"small"} incoherent signal that can be described through a random field. This analysis is beyond the scope of this paper and will not be addressed here.

\section{Proof of Theorem \ref{th:asympt}}\label{proof:th_asympt}

In this section, we describe how the integral operator $\calI$ can be approximated by a fractional derivative, or a standard third order derivative, depending whether $\gamma \in (0,1)$ or $\gamma \geq 1$ respectively. The fractional operator that derives from $\calI$ is given by a Weyl fractional derivative defined by \eqref{def:weyl_dev}. 

In view of \eqref{eq:defI1} and \eqref{eq:defI2}, the operator $\calI$ is completely characterized by the correlation function $R$ through the coefficients $\Gamma_c$ and $\Gamma_s$ (see \eqref{def:Gamma_coef}). Under the scaling \eqref{def:scaleR}, we only need to study the asymptotics to these coefficients, which is done in the following technical lemma, whose proof is provided in Appendix \ref{proof:lemma_coef_scat}.
\begin{lemma}\label{lem:Gamma_scatt}
Setting 
\[ 
\Gamma_c(\omega,l_0):=2\sigma(l_0)\int_0^\infty R\Big(\frac{s}{l_0}\Big) \cos\Big(\frac{2 \omega s}{c_0}\Big) ds,
\]
and
\[
\Gamma_s(\omega,l_0):=2\sigma(l_0)\int_0^\infty R\Big(\frac{s}{l_0}\Big)\sin\Big(\frac{2 \omega s}{c_0}\Big) ds,
\]
where $\sigma(l_0)$ is defined by \eqref{def:sigmal0}, we have:
\begin{enumerate}
\item for $\gamma\in [1,\infty )$,
\[
\lim_{l_0\to 0}\Gamma_c(\omega,l_0) = \Gamma_{0}
\qquad
\text{and}
\qquad
\lim_{l_0\to 0} \Gamma_s(\omega,l_0) = 0,
\]
with $\Gamma_{0}$ defined by \eqref{def:Gamma0};
\item for $\gamma\in(0,1)$,
\[
\lim_{l_0\to 0} \Gamma_c(\omega,l_0) =  R_0 \Gamma_{c,0}(\omega) 
\qquad
\text{and}
\qquad
\lim_{l_0\to 0} \Gamma_s(\omega,l_0) =  R_0  \Gamma_{s,0}(\omega),
\]
with for $\omega\neq 0$
\[
\Gamma_{c,0}(\omega) = 2\int_0^\infty \frac{ds}{s^\gamma} \cos(2\omega s/c_0 ) = 2\Gamma(1-\gamma)\cos\Big(\frac{(1-\gamma)\pi}{2}\Big)\Big(\frac{2|\omega|}{c_0}\Big)^{\gamma-1},
\]
and
\[
 \Gamma_{s,0}(\omega) = 2\int_0^\infty \frac{ds}{s^\gamma} \sin(2\omega s/c_0)= 2\Gamma(1-\gamma)\sin\Big(\frac{(1-\gamma)\pi}{2}\Big)\Big(\frac{2|\omega|}{c_0}\Big)^{\gamma-1}\emph{sign}(\omega),
\]
where $R_0$ is given by \eqref{def:R0}, and $\Gamma$ stands for the Gamma function.
\end{enumerate}
\end{lemma}
Note that the two latter formulas for $\Gamma_{c,0}(\omega)$ and $\Gamma_{s,0}(\omega)$ are obtained using \cite[3.761-4 and 9 pp. 436--437]{tableint}.

From this lemma, we obtain the following convergence in $\calC(\R\times \R^2)$,
\[
\lim_{l_0\to 0} p_{tr,l_0}(s,\by, L) = p_{tr,0}(s,\by, L)=\frac{1}{2}\calK_0(\cdot,\cdot,L) \ast \Psi(s,\by),
\]
thanks to the dominated convergence theorem, where, in the Fourier domain,
\[
\hat \calK_0(\omega,\kappa,z) :=
\left\{ 
\begin{array}{ccc}
\displaystyle e^{-\theta'^2_0 \omega^2 \Gamma_0 z/(8 c_0^2)}  e^{-i\omega c_0|\kappa|^2 z/2} & \text{if} & \gamma \geq 1, \\
&&\\
&&\\
\displaystyle  e^{-\theta'^2_0 R_0 \omega^2 ( \Gamma_{c,0}(\omega)+i  \Gamma_{s,0}(\omega)) z/(8 c_0^2)}  e^{-i\omega c_0|\kappa|^2 z/2} &\text{if} & \gamma \in (0,1).
\end{array}
\right.
\]
From the definition of $\calK_0$, we obtain for $\gamma\geq 1$ the following paraxial wave equation with a third order derivative in $s$,
\[
\partial^2_{s z} \calK_0 + \frac{c_0}{2} \Delta_{\by} \calK_0 - \frac{\theta'^2_0 \Gamma_0}{8 c_0^2} \partial^3_{sss} \calK_0 = 0.
\] 
For $\gamma\in(0,1)$, following the same strategy as in Section \ref{proof:th_main}, but with now
\[
\phi(s)= \frac{\theta'^2_0 R_0 }{2^{3-\gamma}c_0^{1+\gamma}}  s^{-\gamma} \, \1_{(0,\infty)}(s),
\]
we obtain
\[
\partial^2_{sz} \calK_0 + \frac{c_0}{2} \Delta_{\by} \calK_0 - \frac{\theta'^2_0 R_0 \Gamma(1-\gamma)}{2^{3-\gamma}c_0^{1+\gamma}}  D^{2+\gamma}_s\calK_0 = 0,
\] 
where 
\[
D^{2+\gamma}_s\psi(s) = \frac{1}{\Gamma(1-\gamma)} \int_{-\infty}^s (s-\tau)^{-\gamma} \partial^3_{sss} \psi(\tau) d\tau
\]
stands for the Weyl fractional derivative of order $2+\gamma$.

The details for the uniqueness of this paraxial equation are provided in Appendix \ref{sec:uniqueness}.

\section{Proof of Theorem \ref{th:asymp}}\label{proof:th_asymp}

In this section, we adapt the idea of \cite{gomez3} to our context, and the convergence in distribution of the family $(\calX^{\eps})_\eps$ is proved in two steps. The first step concerns the tightness of the family, and the second one concerns the characterization of all its  accumulation points as being weak solutions to a well-posed stochastic differential equation. Both points are based on the perturbed test function method and the notion of pseudogenerator allowing the use of martingale techniques \cite{kushner}. We introduce first this key notion before going into the detailed analysis of the convergence. Let us remark that tightness criteria like \cite[Theorem 4 pp. 48]{kushner} require uniform bounds in probability that we do not have for $(\calX_{\eps})_\eps$. This appears to be also a problem to identify the subsequence limits. To bypass this problem, we adopt the strategy proposed in \cite[Chapter 11]{stroock} and introduce a \emph{truncated process}. This new process is related to the original one through a family of stopping times that goes to $\infty$ when we remove the truncation.

\subsection{The truncated process}

Let us start by writing down the system satisfied by $\calX_{\eps}$ in a convenient form for the forthcoming analysis,
\[
\frac{d}{dz} \calX_{\eps}(z)  = \frac{1}{\eps} \nu_\eps \Big(\frac{z}{\eps}\Big) \sum_{j=1}^n \sum_{l=0,1} \frac{i\omega_j(-1)^l}{2\lambda_\eps(\kappa_j)c^2_0}e^{i\omega_j(-1)^{l+1} \tau_\eps(\kappa_j,z)/\eps}\calX^{j,1-l}_{\eps}(z) \bfe_{jl}.
\]
Here, $\tau_\eps(\kappa,z)$ is given by \eqref{def:tau},
\[
\calX_{\eps}=(\calX^{1,0}_\eps, \calX^{1,1}_\eps,\dots,\calX^{n,0}_\eps,\calX^{n,1}_\eps)=\sum_{j=1}^n\sum_{l=0,1}\calX^{j,l}_{\eps}(z) \bfe_{jl},
\]
with
\[\calX^{j,0}_{\eps}(z)= A_\eps(\omega_j,\kappa_j,z)\qquad\text{and}\qquad\calX^{j,1}_{\eps}(z)= B_\eps(\omega_j,\kappa_j,z)\qquad j\in\{1,\dots,n\},\]
and $(\bfe_{jl})_{jl}$ is the canonical basis for $\mathbb{C}^{2}\times\dots\times\mathbb{C}^{2}$ $n$ times. 

Denoting $M>0$ the cutoff parameter, which does not dependent on $\eps$ in what follows, the truncated process $\calX_{\eps,M}$ is defined as being the solution to
\begin{equation}\label{eq:calX_M}
\frac{d}{dz} \calX_{\eps,M}(z)  = \frac{1}{\eps} \nu_\eps \Big(\frac{z}{\eps}\Big) \sum_{j=1}^n \sum_{l=0,1} \frac{i\omega_j(-1)^l}{2\lambda_\eps(\kappa_j)c^2_0}e^{i\omega_j(-1)^{l+1} \tau_\eps(\kappa_j,z)/\eps}F_M(\calX^{j,1-l}_{\eps,M}(z)) \bfe_{jl},
\end{equation}
where
\[
F_M(\calX):=\calX \, \phi_M(\calX)\qquad \calX\in \mathbb{C},
\]
with $\phi_M$ a compactly supported smooth function such that $0\leq \phi_M \leq 1$, and
\[
\phi_M(\calX)= \left\{
\begin{array}{ccc}
1 & \text{if} & |\calX|\leq M, \\
0 & \text{if} & |\calX|\geq 2M.
\end{array}
\right. 
\]
Thanks to the cutoff function, $\phi_M$ the process $\calX_{\eps,M}$ is uniformly bounded in $\eps$ by $2M$ with probability one. This property is used in the proof of the tightness and the identification of the limiting martingale problem.   

To relate the truncated process $\calX_{\eps,M}$ with the original one $\calX_{\eps}$, let us introduce $\Omega = \calC([0,\infty),\mathbb{C}^{2n})$, the space of all possible trajectories for $\calX_{\eps}$ and $\calX_{\eps,M}$, associated to its canonical filtration
\[\calM_z=\sigma( f(s),\, 0\leq s \leq z  ),\]
and $\sigma$-field
\[
\calM := \sigma\left(\bigcup_{z\geq 0}\calM_z\right).
\]
We also introduce the stopping times
\[
\eta_M(f) := \inf(z\geq 0:\, \|f(z)\|\geq M) \qquad f\in\Omega,
\]
and the distributions $\Pro_\eps$ and $\Pro_{\eps,M}$ of respectively $\calX_{M}$ and $\calX_{\eps,M}$, which are defined on the measurable space $(\Omega,\calM)$. From the above definitions, it is clear that 
\[
\Pro_\eps=\Pro_{\eps,M}\qquad \text{on }\calM_{\eta_M}.
\] 
The strategy of the proof relies on this latter identity together with \cite[Lemma 11.1.1 pp.262]{stroock}. To prove the convergence in distribution, as $\eps\to 0$, of $(\Pro_{\eps})_\eps$ to $\Pro_0$ in $\Omega$, where $\Pro_0$ is the distribution of the process $\calX_0$ defined by \eqref{def:calX0} (in other words the convergence in distribution of $(\calX_\eps)_\eps$ to $\calX_0$), we just have to prove that for each $M>0$:
\begin{itemize}
\item $(\Pro_{\eps,M})_\eps$ (or equivalently $(\calX_{\eps,M})_\eps$) is tight;
\item for any limit point $\Pro_{0,M}$
\[
\Pro_0=\Pro_{0,M}\qquad \text{on }\calM_{\eta_M}.
\] 
\end{itemize}
If the martingale problems associated to $\Pro_0$ and any $\Pro_{0,M}$ are the same on $\calM_{\eta_M}$, and $\Pro_0$ is associated to a well-posed stochastic differential equation, the latter point is a direct consequence of \cite[Exercise 11.5.1 pp.283]{stroock}. The identification of the limiting martingale problem is carried out in Section \ref{sec:identification}, while the tightness is proved in Section \ref{tightnesssec}. Also, to prove the tightness of $(\Pro_{\eps,M})_\eps$ on $\Omega$, we only have to prove it on $\calC([0,L],\mathbb{C}^{2n})$ for any $L>0$. 

From the boundedness of the truncated process $\calX_{\eps,M}$, together with the expansion 
\begin{equation}\label{eq:expensionV}
\nu_\eps(z/\eps) = \Theta\big(\seps V(z/\eps)\big) =  \seps \Theta'(0) V(z/\eps) + \calO\big( \eps^{3/2} |V(z/\eps)|^3 \big),
\end{equation}
remembering that $\Theta$ is an odd function, and \eqref{eq:calX_M}, one can write
\begin{align*}
\frac{d}{dz} \calX_{\eps,M}(z) &= \frac{\theta'_0}{\seps}  V \Big(\frac{z}{\eps}\Big) \sum_{j=1}^n\sum_{l=0,1} \frac{i\omega_j(-1)^l}{2\lambda_\eps(\kappa_j)c^2_0}e^{i\omega_j(-1)^{l+1} \tau_\eps(\kappa_j,z)/\eps}F_M(\calX^{j,1-l}_{\eps,M}(z)) \bfe_{jl} \\
&+ E_\eps(z).
\end{align*}
Here, $\theta'_0 = \Theta'(0)$, and the error term $E_\eps$ is uniformly bounded in $\eps$ and $z\in[0,L]$ with probability one. This latter term provides a negligible contribution in the limit $\eps\to 0$ and does not play any role in the forthcoming analysis. For the sake of simplicity in the presentation, we ignore this term and consider instead the following system:
\begin{equation}\label{eq:calX_M2}
\frac{d}{dz} \calX_{\eps,M}(z) = \frac{\theta'_0}{\seps}  V \Big(\frac{z}{\eps}\Big)  \sum_{j=1}^n \sum_{l=0,1} \frac{i\omega_j(-1)^l}{2\lambda_\eps(\kappa_j)c^2_0}e^{i\omega_j(-1)^{l+1} \tau_\eps(\kappa_j,z)/\eps}F_M(\calX^{j,1-l}_{\eps,M}(z)) \bfe_{jl}.
\end{equation}

\subsection{Pseudogenerator}

We remind the reader about the notion of pseudogenerator allowing the use of martingale techniques while the underlying process is not a Markov process. Before introducing the notion of pseudogenerator, let us defined the $p-\lim$.

Let us introduce the following $\sigma$-algebras
\[
\mathcal{G}^\eps _z = \sigma(V(s/\eps,dp),\quad 0\leq s\leq z)\qquad 0\leq z\leq L,
\]
and $\calS^\eps $ be the set of all measurable functions $f$, adapted to the filtration $(\mathcal{G}^\eps _z)$, and for which $\sup_{z\leq L} \E[\lvert f(z) \rvert ] < \infty $.  Let $f$ and $f_h$ in $\mathcal{S}^\eps $ for all $h>0$, we say that $f=p-\lim_h f_h$ if
\[
\sup_{z, h }\E[\lvert f_h(z)\rvert]<+\infty\qquad \text{and}\qquad \lim_{h\rightarrow 0}\E[\lvert f_h(z) -f(z)\rvert]=0 \qquad \forall z\geq 0.
\]

Regarding the pseudogenerator itself, denoted by $\mathcal{A}^\eps$, we say that $f\in \mathcal{D}(\mathcal{A}^\eps)$, the domain of $\mathcal{A}^\eps$, and $\mathcal{A}^\eps f=g$ if both $f$ and $g$ are in $\mathcal{S}^\eps$ and 
\[
p-\lim_{h \to 0} \left[ \frac{\E^\eps _z [f(z+h)]-f(z)}{h}-g(z) \right]=0.
\]
Here, $\E^\eps _z$ denotes the conditional expectation given $\mathcal{G}^\eps _z$. The key property to relate the pseudogenerator to the martingale property is the following.

\begin{proposition}\label{martingale}
For any $f\in \mathcal{D}(\mathcal{A}^\eps)$, the process
\begin{equation*}
M_f ^\eps (z)=f(z)-f(0)-\int _0 ^z  \mathcal{A}^\eps f(u)du
\end{equation*}
is a $( \mathcal{G}^\eps _z )$-martingale.
\end{proposition} 

This last result will allow us to characterize the limiting process of $\calX_{\eps,M}$ through a well-posed martingale problem with generator $\calA$ that has to be determined. Unfortunately, the pseudogenerator $\calA^\eps$ associated to $\calX_{\eps,M}$, at some test function $f$, has a singular term of order $1/\sqrt{\eps}$. The idea of the perturbed test function method is to construct a perturbation $f^\eps$ of $f$ in order to extract an effective statistical behavior from $\calA^\eps f^\eps$. This strategy allows $\calA^\eps f^\eps$ to converge to $\calA f$, where $\calA$ will be the generator describing this asymptotic statistical behavior for $\calX_{0,M}$ (with distribution $\Pro_{0,M}$).

\subsection{Technical lemmas for the fluctuations $V$}

Here, we introduce two results that are used in the forthcoming analysis to analyze the corrections of a test function.  We refer to \cite[Appendices C and D]{gomez3} for detailed proofs of these two lemmas. The first one concerns the conditional expectation and variance.

\begin{lemma}\label{l:mar}
Setting 
\[
\mathcal{G}_z = \sigma(V(s,dp),\quad 0\leq s\leq z)\qquad z\geq 0,
\]
where $V(\cdot,dp)$ is given by \eqref{def:Vdp},  we have for any $z, h\geq 0$
\begin{equation}\label{eq:markovesp}
\E\big[ V(z+h,dp) \vert \mathcal{G}_{z}\big]=e^{-\mu\vert p\vert^{2\beta}h}\,V(z,dp),
\end{equation}
and
\begin{equation}\label{eq:markovvar}\begin{split}
\E\big[V(z+h,dp)V(z+h,dq) \big\vert \mathcal{G}_{z}\big]&-  \E\big[V(z+h,dp) \vert \mathcal{G}_{z}\big]\E\big[V(z+h,dq)\vert \mathcal{G}_{z}\big] \\
&=(1-e^{-2\mu\vert p\vert^{2\beta}h})r(p)\delta(p-q)\,dp\,dq\,.
\end{split}\end{equation}
\end{lemma}
The second result concerns uniform bounds for the fluctuations $V$.   
\begin{lemma}\label{l:bound} Let $L>0$, $M>0$,
  \begin{equation*} 
  D_{k,M} := [0,L]\times L^\infty([0,L],W_{k,M}), 
  \end{equation*}
  with
  \begin{equation*}
  W_{k,M} := \big\{\varphi\in W^{1,k}(S):\quad \|\varphi\|_{W^{1,k}}\leq M \big\},
  \end{equation*}
where $W^{1,k}(S)$ stands for the Sobolev space with $k\in(1,\infty]$. We have
  \begin{equation}\label{eq:prop1V}\E\Big[\sup_{(z,\varphi)\in D_{k,M}} \Big\vert V\Big(\frac{z}{\eps},\varphi(z,\cdot)\Big)\Big\vert\Big]\leq C + \frac{C(\eps)}{\seps}\,,\end{equation}
  and for any $n\in\mathbb{N}^*$
  \begin{equation}\label{eq:prop2V}\sup_{\eps}\sup_{z\in[0,L]}\E\Big[\sup_{\varphi\in W_{k,M}} \Big\vert V\Big(\frac{z}{\eps},\varphi\Big)\Big\vert ^n\Big]\leq C_n\,,\end{equation}
  where $C$, $C_n$ and $C(\eps)$ are three positive constants, and the latter satisfies 
  \[\lim_{\eps\to 0} C(\eps) = 0.\] 
\end{lemma}

\subsection{Tightness}\label{tightnesssec}

In this section, we prove the tightness of $(\calX_{\eps,M})_{\eps }$, which is a family of processes with continuous trajectories. According to \cite[Theorem 13.4]{billingsley}, it is enough to prove its tightness in $\mathcal{D}([0,L],\mathbb{C}^{2n})$, the set of c\`ad-l\`ag functions with values in $\mathbb{C}^{2n}$, and equipped with the Skorohod topology.
	
\begin{proposition}
The family $(\calX_{\eps,M})_{\eps}$ is tight in $\mathcal{D}([0,L],\mathbb{C}^{2n})$.
\end{proposition}
This proposition can be proved using the perturbed test function method by applying \cite[Theorem 4 pp. 48]{kushner}. Throughout the forthcoming analysis, we make use of the complex derivatives that are defined, for $\calX=u+iv$, as
\[
\partial_\calX := \frac{1}{2}(\partial_u - i\partial_v)\qquad \text{and}\qquad \partial_{\overline{\calX}} := \frac{1}{2}(\partial_u + i\partial_v).
\]
These tools allow us to keep working with complex quantities, and avoid working with \eqref{eq:calX_M2} rewritten in terms of real and imaginary parts.

In what follows, let $f$ be a smooth bounded function on $\mathbb{C}^{2n}$ with successive bounded derivatives, and set 
\[f_0 ^\eps (z):=f(\calX_{\eps,M}(z)).\]
In order to prove the tightness, we make use of the pseudogenerator and associated martingale techniques. The pseudogenerator for $\calX_{\eps,M}$ at $f_0^\eps$ is given by  
\begin{align*}
\mathcal{A}^\eps f_0 ^\eps (z) &= \frac{\theta'_0}{\seps}  V \Big(\frac{z}{\eps}\Big) \\
&\times \Big( \sum_{j=1}^n \sum_{l=0,1} \frac{i\omega_j(-1)^l}{2\lambda_\eps(\kappa_j)c^2_0}e^{i\omega_j(-1)^{l+1} \tau_\eps(\kappa_j,z)/\eps}F_M(\calX^{j,1-l}_{\eps,M}(z))  \partial_{\calX^{j,l}} f(\calX_{\eps,M}(z))\\
&+ \overline{\sum_{j=1}^n \sum_{l=0,1} \frac{i\omega_j(-1)^l}{2\lambda_\eps(\kappa_j)c^2_0}e^{i\omega_j(-1)^{l+1} \tau_\eps(\kappa_j,z)/\eps}F_M(\calX^{j,1-l}_{\eps,M}(z)) }\, \partial_{\overline{\calX^{j,l}}} f(\calX_{\eps,M}(z))\Big)\\
& =: \frac{\theta'_0}{\seps}  V \Big(\frac{z}{\eps}\Big) \\
&\times \Big(  \sum_{j=1}^n \sum_{l=0,1} \frac{i\omega_j(-1)^l}{2\lambda_\eps(\kappa_j)c^2_0}e^{i\omega_j(-1)^{l+1} \tau_\eps(\kappa_j,z)/\eps}F_M(\calX^{j,1-l}_{\eps,M}(z))  \partial_{\calX^{j,l}} f(\calX_{\eps,M}(z))\\
& \hspace{3cm}+c.c.\Big),
\end{align*}
where $c.c.$ stands for \emph{complex conjugate}, and will be used throughout the remaining of this proof instead of rewriting quantities that only have to be conjugated.

The tightness of $(\calX_{\eps,M})_{\eps }$ is proved through \cite[Theorem 4 pp. 48]{kushner} by Lemmas \ref{bound2} and \ref{A1} below. The main tool behind these technical requirements is the martingale property provided by Proposition \ref{martingale}, which involves the pseudogenerator $\calA^\eps$. Therefore, to prove the tightness, we need to remove the singular term produced by $V(\cdot/\eps)/\seps$ in $\mathcal{A}^\eps f_0 ^\eps$, and then in the corresponding martingale itself. To this end, we construct a small perturbation $f_1^{\eps}$ to $f_0^{\eps}$ (Lemma \ref{bound2}) so that the pseudogenerator $\mathcal{A}^\eps(f^\eps _0 +f^\eps _1)$ will become of order one w.r.t. $\eps$ (Lemma \ref{A1}), as well as the martingale associated to $f^\eps _0 +f^\eps _1$ in Proposition \ref{martingale}. 

Following the strategy of \cite[Chapter 4]{kushner}, we set
\begin{align*}
f^\eps _{1} (z) = \frac{\theta'_0}{\seps} & \int_{z}^{\infty}  ds \, \E^\eps_z \Big[V\Big(\frac{s}{\eps}\Big)\Big] \\
&\times \Big(\sum_{j=1}^n\sum_{l=0,1} \frac{i\omega_j(-1)^l}{2\lambda_\eps(\kappa_j)c^2_0}e^{i\omega_j(-1)^{l+1} (\tau_\eps(\kappa_j,z)+2\lambda_\eps(\kappa_j)(s-z))/\eps}\\
&\hspace{5cm}\times F_M(\calX^{j,1-l}_{\eps,M}(z))  \partial_{\calX^{j,l}} f(\calX_{\eps,M}(z))\\
&\hspace{1cm}+ c.c. \Big),
\end{align*}
for which we have the two following results.
\begin{lemma}\label{bound2}  For any $\eta>0$
\[
\lim_{\eps\to 0} \Pro\Big(\sup_{z \in[0, L]} \lvert f^\eps _1 (z)\rvert>\eta\Big) =0,\qquad \text{and} \qquad \lim_{\eps\to 0} \sup_{z \in[0, L]}\E[\lvert f^\eps _1 (z) \rvert ]=0.
\]
\end{lemma}
\begin{lemma}\label{A1}
The family $\big\{\mathcal{A}^\eps \big(f^\eps _0 +f^\eps _1\big)(z), \eps \in(0,1), 0\leq z\leq L\big\}$ is uniformly integrable.
\end{lemma}

\begin{proof}[Proof of Lemma \ref{bound2}]
According to \eqref{eq:markovesp} we have
\[
\E^\eps_z \Big[V\Big(\frac{z}{\eps}+s\Big)\Big] = \int_{S} e^{-\fg(p)s}V\Big(\frac{z}{\eps},dp\Big),
\]
where we have introduced the notation
\[\fg(p)=\mu |p|^{2\beta}\]
for simplicity. Making the change of variable $s\to z+\eps s$, and integrating in $s$, we obtain
\begin{align*}
f^\eps _{1} (z) & =  \seps \, \theta'_0 \sum_{j=1}^n \sum_{l=0,1} \int  \frac{V(z/\eps,dp)}{\fg(p) + 2i\omega_j (-1)^{l+1}  \lambda_\eps(\kappa_j)} \\
&  \hspace{2cm}\times \frac{i\omega_j(-1)^l}{2\lambda_\eps(\kappa_j)c^2_0}e^{i\omega_j(-1)^{l+1} \tau_\eps(\kappa_j,z)/\eps}F_M(\calX^{j,1-l}_{\eps,M}(z)) \partial_{\calX^{j,l}} f(\calX_{\eps,M}(z))\\ 
&+c.c..
\end{align*}
It turns out that
\[
|f^\eps _{1} (z)| \leq \seps \, K \sum_{j=1}^n  \Big|\int_S \frac{V(z/\eps,dp)}{\fg(p) + 2i\omega_j \lambda_\eps(\kappa_j)} \Big| + \Big|\int_S \frac{V(z/\eps,dp)}{\fg(p) - 2i\omega_j \lambda_\eps(\kappa_j)} \Big|, \]
for some constant $K>0$. Let us denote
\begin{equation}\label{def:Vphi1}
V(s,\varphi_{1,j,\eps}) = \int_S \frac{V(s,dp)}{\fg(p) \pm 2i\omega_j \lambda_\eps(\kappa_j)}, \qquad\text{where} \qquad \varphi_{1,j,\eps}(p):= \frac{1}{\fg(p) \pm 2i\omega_j \lambda_\eps(\kappa_j)},
\end{equation}
which is Lipschitz in $p$ if $\beta\geq 1/2$, or belongs to $W^{1,k}(S)$ for $k\in(1,1/(1-2\beta))$ if $\beta<1/2$ since
\[
\int_S |\partial_p \varphi_{1,j,\eps}(p)|^k dp \leq \tilde C \int |p|^{k(2\beta-1)} <\infty.
\]
Therefore, $\varphi_{1,j,\eps}\in W_{k,C}$ for some constant $C>0$ and
\[
|V(z/\eps,\varphi_{1,j,\eps})| \leq \sup_{\varphi \in W_{k,C}}|V(z/\eps,\varphi)|,
\]
so that
\[
|f^\eps _{1} (z)| \leq K' \seps \sup_{\varphi\in W_{k,C}} |V(z/\eps,\varphi)|.
\]
Consequently, according to \eqref{eq:prop2V} we obtain 
\[
\lim_{\eps\to 0} \sup_{z \in[0, L]}\E[\lvert f^\eps _1 (z) \rvert ]=0.
\]
Also, we have for any $\eta'>0$
\begin{align*}
\Pro\Big(\sup_{z \in[0, L]} \lvert f^\eps _1 (z)\rvert>\eta\Big)& \leq \Pro\Big(\sup_{z \in[0, L]} \lvert f^\eps _1 (z)\rvert>\eta,\, \seps \sup_{(s,\varphi)\in [0,L/\eps]\times W_{k,C}} |V(s,\varphi)|\leq \eta' \Big)\\
& + \Pro\Big( \seps \sup_{(s,\varphi)\in [0,L/\eps]\times W_{k,C}} |V(s,\varphi)|>\eta' \Big)=: P_{1,\eps}+P_{2,\eps},
\end{align*}
where $P_{1,\eps}=0$ since $K'\eta'\leq \eta$ for $\eta'$ small enough but independent of $\eps$. Finally, according to \eqref{eq:prop1V}, together with the Markov inequality, we obtain
\[0\leq P_{2,\eps} \leq \frac{1}{\eta'}\E\Big[\seps \sup_{(s,\varphi)\in [0,L/\eps]\times W_{k,C}} |V(s,\varphi)|\Big],\]
so that $\lim_{\eps\to 0} P_{2,\eps}=0$, which concludes the proof of the lemma.
\end{proof}

\begin{proof}[Proof of Lemma \ref{A1}]
After lengthy but straightforward algebra, we obtain 
\begin{equation}\label{eq:exp_A}
\mathcal{A}^\eps (f_0 ^\eps + f^\eps _{1}) (z) = \mathcal{A}^\eps_0(z)+ \mathcal{A}^\eps_1(z),
\end{equation}
where
\begin{align*}
\mathcal{A}^\eps_0(z) &:=  \sum_{j=1}^n\sum_{l=0,1} \theta'^2_0  \int \frac{V (z/\eps) V(z/\eps,dp)}{\fg(p) + 2i\omega_j (-1)^{l+1}  \lambda_\eps(\kappa_j)}\\
&\times \Big( \sum_{j'=1}^n\sum_{l'=0,1} \bF_{1,j,j',l,l',\eps}(z) e^{i( \omega_j(-1)^{l+1} \tau_\eps(\kappa_j,z)+\omega_{j'}(-1)^{l'+1} \tau_\eps(\kappa_{j'},z))/\eps}\\
&\hspace{2cm} + \bF_{2,j,j',l,l',\eps}(z) e^{i( \omega_j(-1)^{l+1} \tau_\eps(\kappa_j,z)-\omega_{j'}(-1)^{l'+1} \tau_\eps(\kappa_{j'},z))/\eps}\\
& \hspace{0.5cm} + \bG_{1,j,\eps}(z) + \bG_{2,j,\eps}(z) e^{2i \omega_j(-1)^{l+1} \tau_\eps(\kappa_j,z)/\eps}\Big)+ c.c.,\\
\mathcal{A}^\eps_1(z) &:= \sum_{j=1}^n\sum_{l=0,1} \frac{\theta'_0}{\seps} \int \frac{ \Theta(\seps V(z/\eps))V(z/\eps,dp)}{\fg(p) + 2i\omega_j (-1)^{l+1}  \lambda_\eps(\kappa_j)}
\bH_{j,\eps}(z)e^{i\omega_j(-1)^{l+1} \tau_\eps(\kappa_j,z)/\eps}+ c.c.,
\end{align*} 
with
\begin{align*}
\bF_{1,j,j',l,l',\eps}(z)&:=  \frac{i^2\omega_j\omega_{j'}(-1)^{l+l'}}{4\lambda_\eps(\kappa_j)\lambda_\eps(\kappa_{j'})c^4_0} F_M(\calX^{j,1-l}_{\eps,M}(z))F_M(\calX^{j',1-l'}_{\eps,M}(z)) \\
&\hspace{6cm}\times \partial^2_{\calX^{j',l'}\calX^{j,l}} f(\calX_{\eps,M}(z)),\\
\bF_{2,j,j',l,l',\eps}(z)&:=  \frac{-i^2\omega_j\omega_{j'}(-1)^{l+l'}}{4\lambda_\eps(\kappa_j)\lambda_\eps(\kappa_{j'})c^4_0} F_M(\calX^{j,1-l}_{\eps,M}(z))\overline{F_M(\calX^{j',1-l'}_{\eps,M}(z))} \\
&\hspace{6cm}\times \partial^2_{\overline{\calX^{j',l'}}\calX^{j,l}} f(\calX_{\eps,M}(z)),\\
\bG_{1,j,l,\eps}(z)&:= \frac{-i^2\omega^2_j}{4\lambda^2_\eps(\kappa_j)c^4_0}F_M(\calX^{j,l}_{\eps,M}(z)) \partial_{\calX} F_M(\calX^{j,l}_{\eps,M}(z))\partial_{\calX^{j,l}} f(\calX_{\eps,M}(z)), \\
\bG_{2,j,l,\eps}(z)&:=\frac{i^2\omega^2_j}{4\lambda^2_\eps(\kappa_j)c^4_0}\overline{F_M(\calX^{j,l}_{\eps,M}(z))} \partial_{\overline{\calX}} F_M(\calX^{j,l}_{\eps,M}(z)) \partial_{\calX^{j,l}} f(\calX_{\eps,M}(z)),\\
\bH_{j,l,\eps}(z)&:= \frac{-i^2\omega^2_j}{4\lambda^2_\eps(\kappa_j)c_0^4} F_M(\calX^{j,l}_{\eps,M}(z))\partial_{\calX^{j,l}} f(\calX_{\eps,M}(z)).
\end{align*}

Note that both $\mathcal{A}^\eps_0$ and $\mathcal{A}^\eps_1$ depend on the original test function $f$, even if we drop this dependency for notational simplicity. To conclude the proof of the tightness, we can see that both $\mathcal{A}^\eps_0$ and $\mathcal{A}^\eps_1$ are uniformly integrable thanks to the following lemma.

\begin{lemma}\label{lem:bound_As}
We have
\[
\sup_{\eps,z\in[0,L]} \E[|\mathcal{A}^\eps_0(z)|^2] + \E[|\mathcal{A}^\eps_1(z)|^2] <\infty.
\]
\end{lemma}
\begin{proof}[Proof of Lemma \ref{lem:bound_As}]
Let us treat only the term $\mathcal{A}^\eps_1$ involving a term $\Theta(\seps V)$. The treatment of $\mathcal{A}^\eps_0$ follows the same lines once we get ride of the function $\Theta$ in $\mathcal{A}^\eps_1$.
From \eqref{eq:expensionV}, we have
\begin{equation}\label{eq:A1}\begin{split}
\mathcal{A}^\eps_1(z) & = \sum_{j=1}^n\sum_{l=0,1} \theta'^2_0 \int \frac{ V(z/\eps)V(z/\eps,dp)}{\fg(p) + 2i\omega_j (-1)^{l+1}  \lambda_\eps(\kappa_j)}
\bH_{j,\eps}(z)e^{i\omega_j(-1)^{l+1} \tau_\eps(\kappa_j,z)/\eps}+ c.c.\\
&+\calO\Big(\eps |V(z/\eps)|^3 \sum_{j=1}^n |V(z/\eps,\varphi_{1,j,\eps})|\Big),
\end{split}
\end{equation}
where $\varphi_{1,j,\eps}$ is defined by \eqref{def:Vphi1}. Following the lines of the proof of Lemma \ref{bound2}, we obtain 
\begin{equation*}
|\mathcal{A}^\eps_1(z)|^2 \leq K \Big( \sup_{\varphi\in W_{k,C}} |V(z/\eps,\varphi)|^4 + \eps^2 \sup_{\varphi\in W_{k,C}} |V(z/\eps,\varphi)|^8\Big),
\end{equation*}
for some appropriate positive constants $K$ and $C$, with $k=\infty$ if $\beta \geq 1/2$, and $k\in (1,1/(1-2\beta))$ if $\beta<1/2$. This concludes the proof of Lemma \ref{lem:bound_As},
\end{proof}
and then the one of Lemma \ref{A1} owing \eqref{eq:exp_A}.
\end{proof}

\subsection{Identification of the limit}\label{sec:identification}

In this section we identify all the limit points of $(\calX_{\eps,M})_\eps$ through a martingale problem with infinitesimal generator
\begin{equation}\label{def:LM}\begin{split}
\calL_M f(\calX) = \sum_{j=1}^n\sum_{l=0,1} \frac{\theta'^2_0 \omega^2_j}{4c^2_0}  \int & \frac{r(p)}{\fg(p) + 2i(\omega_j/c_0) (-1)^{l+1}}\\
&\times \Big( F_M(\calX^{j,1-l})F_M(\calX^{j,l}) \, \partial^2_{\calX^{j,1-l}\calX^{j,l}} f(\calX)\\
& +  F_M(\calX^{j,1-l}_{\eps,M})\overline{F_M(\calX^{j,1-l})} \, \partial^2_{\overline{\calX^{j,l}}\calX^{j,l}} f(\calX ) \\
&+ F_M(\calX^{j,l}) \partial_{\calX} F_M(\calX^{j,l})\partial_{\calX^{j,l}} f(\calX)\Big)\\
& \hspace{-2cm} + c.c..
\end{split}
\end{equation}
Let us start with the following remark. In view of \eqref{eq:A1} together with \eqref{eq:prop2V}, we have
\[
\sup_{\eps,z\in [0,L]} \E[|\calA^\eps_1(z) - \calA'^\eps_1(z)|]=\calO(\eps), 
\]
with
\[
\calA'^\eps_1(z):=  \sum_{j=1}^n\sum_{l=0,1} \theta'^2_0 \int \frac{ V(z/\eps)V(z/\eps,dp)}{\fg(p) + 2i\omega_j (-1)^{l+1}  \lambda_\eps(\kappa_j)}
\bH_{j,l,\eps}(z)e^{i\omega_j(-1)^{l+1} \tau_\eps(\kappa_j,z)/\eps}+ c.c..
\]
As a result, remembering \eqref{eq:exp_A}, we have
\[
\sup_{z\in [0,L]} \E[|\calA^\eps(f_0^\eps + f_1^\eps)(z) - \calA^\eps_0(z) - \calA'^\eps_1(z)|] = \calO(\eps),
\]
so that to determine the infinitesimal generator of the limit points we only have to focus on $\calA^\eps_0 + \calA'^\eps_1$. For this term, we separate the terms which exhibit a fast phase from the others, since the former will average out and do not contribute at the limit $\eps\to 0$. We then write
\[
\calA^\eps_0(z) + \calA'^\eps_1(z) = \calB^\eps_0(z) + \calB^\eps_1(z),
\]
where
\begin{align*}
\calB^\eps_0(z) &:=\sum_{j=1}^n\sum_{l=0,1} \theta'^2_0  \int \frac{V (z/\eps) V(z/\eps,dp)}{\fg(p) + 2i\omega_j (-1)^{l+1}  \lambda_\eps(\kappa_j)} \\
&\hspace{3cm}\times \Big( \bF_{1,j,j,l,1-l,\eps}(z) + \bF_{2,j,j,l,l,\eps}(z) + \bG_{1,j,l,\eps}(z)\Big)\\
&+ c.c.,
\end{align*}
and
\begin{align*}
\calB^\eps_1(z) := \calA^\eps_1(z)+\sum_{j=1}^n & \sum_{l=0,1}  \theta'^2_0  \int  \frac{V (z/\eps) V(z/\eps,dp)}{\fg(p) + 2i\omega_j (-1)^{l+1}  \lambda_\eps(\kappa_j)}\\
&\times \Big( \sum_{\substack{j\neq j' \\ \text{or} \\ l=l'}} \bF_{1,j,j',l,l',\eps}(z) e^{i( \omega_j(-1)^{l+1} \tau_\eps(\kappa_j,z)+\omega_{j'}(-1)^{l'+1} \tau_\eps(\kappa_{j'},z))/\eps}\\
& +\sum_{\substack{j\neq j' \\ \text{or} \\ l=1-l'}}\bF_{2,j,j',l,l',\eps}(z) e^{i( \omega_j(-1)^{l+1} \tau_\eps(\kappa_j,z)-\omega_{j'}(-1)^{l'+1} \tau_\eps(\kappa_{j'},z))/\eps}\\
& \hspace{0.5cm} + \bG_{2,j,l,\eps}(z) e^{2i \omega_j(-1)^{l+1} \tau_\eps(\kappa_j,z)/\eps}\Big)+ c.c..
\end{align*}

\subsubsection{The term $\calB^\eps_1$}

Because of its rapid phases, this term does not contribute to the limit. To prove this, we start by introducing another test function to average out the stochastic terms involving $V$. For notational convenience, and without loss of generality, we only treat the term involving $\calA'^\eps_1$. The other terms are treated exactly the same way. 

Setting 
\begin{align*}
f^\eps_2(z) := \int_z^\infty ds\,\sum_{j=1}^n\sum_{l=0,1} \theta'^2_0 & \iint \frac{ \E_z^\eps[V(s/\eps,dp)V(s/\eps,dq)] - \E[V(0,dp)V(0,dq)] }{\fg(p) + 2i\omega_j (-1)^{l+1}  \lambda_\eps(\kappa_j)}\\
&\times \bH_{j,l,\eps}(z)e^{i\omega_j(-1)^{l+1} ( \tau_\eps(\kappa_j,z) + 2\lambda_\eps(\kappa_j)(s-z)) /\eps},
\end{align*}
we have
\[
\calA^\eps(f^\eps_2)(z)=-\calA'^\eps_1(z) +\calA''^\eps_1(z)+ \seps \, \calR^\eps_1(z),
\]
where
\[
\calA''^\eps_1(z) := \sum_{j=1}^n\sum_{l=0,1} \theta'^2_0 \int \frac{ r(p)}{\fg(p) + 2i\omega_j (-1)^{l+1}  \lambda_\eps(\kappa_j)}
\bH_{j,l,\eps}(z)e^{i\omega_j(-1)^{l+1} \tau_\eps(\kappa_j,z)/\eps} +c.c.,
\]
and the following lemma.
\begin{lemma}\label{bound:f2} We have
\[
\sup_{z\in[0,L]}\E[|f^\eps_2(t)|]=\calO(\eps)\qquad\text{and}\qquad \sup_{\eps,z\in[0,L]}\E[| \calR^\eps_1(z)|]<\infty.
\]
\end{lemma}
\begin{proof}[Proof of Lemma \ref{bound:f2}]
Making the change of variable $s\to z + \eps s$, together with \eqref{eq:markovvar}, and integrating in $s$, yield
\begin{align*}
f_2^\eps(z) = \eps \sum_{j=1}^n&\sum_{l=0,1} \theta'^2_0 \bH_{j,l,\eps}(z)e^{i\omega_j(-1)^{l+1}  \tau_\eps(\kappa_j,z)  /\eps} \\
& \times \iint \frac{ V(z/\eps,dp)V(z/\eps,dq) - r(p)\delta(p-q)dpdq }{(\fg(p) + 2i\omega_j (-1)^{l+1}  \lambda_\eps(\kappa_j))(\fg(p)+\fg(q) + 2i\omega_j (-1)^{l+1}  \lambda_\eps(\kappa_j))}\\
&=: f_{21}^\eps(z) + f_{22}^\eps(z),
\end{align*}
where $f_{21}^\eps$ corresponds to the term with $V(z/\eps, dp)V(z/\eps,dq)$ in $f_{2}^\eps$, and $f_{22}^\eps$ the one with $r(p)$. It is direct to see that
\[
\sup_{z\in[0,L]}\E[|f_{22}^\eps(z)|] = \calO(\eps),
\]
and we only need to focus on $f_{21}^\eps$.
 
Now, let us denote
\[
V(s,\varphi_{2,j,s,\eps}) = \int_S\int_S \frac{V(s,dp)V(s,dq)}{\big(\fg(p) \pm 2i\omega_j \lambda_\eps(\kappa_j)\big)\big(\fg(p)+\fg(q) \pm 2i\omega_j \lambda_\eps(\kappa_j)\big)},
\]
where
\[
\varphi_{2,j,s,\eps}(p):= V(s,\varphi_{3,j,p,\eps}),
\]
and
\[
\varphi_{3,j,p,\eps}(q):=\frac{1}{\big(\fg(p) \pm 2i\omega_j \lambda_\eps(\kappa_j)\big)\big(\fg(p)+\fg(q) \pm 2i\omega_j \lambda_\eps(\kappa_j)\big)}.
\]
In the same way as for $\varphi_{1,j,\eps}$ in the proof of Lemma \ref{bound2}, we can remark that for any $p$ the function $\varphi_{3,j,p,\eps}$ belongs to $W_{k,C'}$ for some constant $C'>0$ independent to $p$, and so that
\[
\sup_{p\in S} |\varphi_{2,j,s,\eps}(p)| \leq \sup_{\varphi\in  W_{k,C'}} |V(s,\varphi)|.
\]
Also, we have
\[
\partial_p \varphi_{2,j,s,\eps}(p) = V(s,\tilde\varphi_{3,j,p,\eps}),
\]
with $\tilde\varphi_{3,j,p,\eps}$ belonging to $W_{k,C''}$, where $C''>0$ does not depend on $p$, and then
\[
\sup_{p\in S} |\partial_p \varphi_{2,j,s,\eps}(p)| \leq \sup_{\varphi\in W_{k,C''}} |V(s,\varphi)|.
\]
Considering $C''' = \max(C',C'')$, we have
\[
\|\varphi_{2,j,s,\eps}\|_{W^{1,k}(S)}\leq \sup_{\varphi\in W_{k,C'''}} |V(s,\varphi)|,
\]
and setting 
\[
\tilde \varphi_{2,j,s,\eps}:= \frac{\varphi_{2,j,s,\eps}}{\sup_{\varphi\in W_{k,C'''}} |V(s,\varphi)|},
\]
we have $\|\tilde \varphi_{2,j,s,\eps}\|_{W^{1,k}(S)}\leq 1$, and then $\tilde \varphi_{2,j,s,\eps}\in W_{k,1}$. As a result, we obtain
\[
|V(z/\eps, \varphi_{2,j,z/\eps,\eps})| = \sup_{\varphi\in W_{k,C'''}} |V(z/\eps,\varphi)| \, |V(z/\eps, \tilde \varphi_{2,j,z/\eps,\eps})| \leq \sup_{\varphi\in W_{k,C_1}} |V(z/\eps,\varphi)|^2,
\]
with $C_1=\max(1,C''')$, so that
\[
\sup_{z\in[0,L]}\E[|f_{21}^\eps(z)|] = \calO(\eps),
\]
according to \eqref{eq:prop2V}. 

The term $\calR^\eps_1(z)$ is treated in a similar way and we omit the precise details. This concludes the proof of Lemma \ref{bound:f2}.
\end{proof} 

Now, to deal with the rapid phases of $\calA''^\eps_1$, we introduce the following test function
\begin{align*}
f^\eps_3(z) := \int_z^\infty ds\,\sum_{j=1}^n\sum_{l=0,1} \theta'^2_0 & \iint \frac{ r(p) }{\fg(p) + 2i\omega_j (-1)^{l+1}  \lambda_\eps(\kappa_j)}\\
&\times \bH_{j,l,\eps}(z)e^{i\omega_j(-1)^{l+1} ( \tau_\eps(\kappa_j,z) + 2\lambda_\eps(\kappa_j)(s-z)) /\eps} e^{-\seps (s-z)},
\end{align*}
so that making again the change of variable $s\to z+\eps s$ and integrating in $s$, we have
\[
f^\eps_3(z) := \eps\sum_{j=1}^n\sum_{l=0,1} \theta'^2_0  \iint \frac{ r(p)\bH_{j,l,\eps}(z)e^{i\omega_j(-1)^{l+1}  \tau_\eps(\kappa_j,z)/\eps} }{(\eps^{3/2} - 2i\omega_j (-1)^{l+1}  \lambda_\eps(\kappa_j))(\fg(p) + 2i\omega_j (-1)^{l+1}  \lambda_\eps(\kappa_j))},
\]
satisfying
\[
\sup_{z\in[0,L]}\E[|f_3^\eps(z)|] = \calO(\eps).
\]
Now, differentiating in $z$, we obtain
\[
\calA^\eps(f_3^\eps)(z) = - \calA''^\eps_0(z) +  \calR^\eps_2(z),
\]
with
\begin{align*}
\calR^\eps_2(z) := \int_z^\infty ds\,\sum_{j=1}^n & \sum_{l=0,1} \theta'^2_0  \iint \frac{ r(p) }{\fg(p) + 2i\omega_j (-1)^{l+1}  \lambda_\eps(\kappa_j)}\\
&\times e^{i\omega_j(-1)^{l+1} ( \tau_\eps(\kappa_j,z) + 2\lambda_\eps(\kappa_j)(s-z)) /\eps} e^{-\seps (s-z)}\\
&\times \Big(  \frac{d}{dz}\bH_{j,l,\eps}(z) + \frac{i\omega_j(-1)^{l+1}}{2\eps\lambda_\eps(\kappa_j)c_0^2}\bH_{j,l,\eps}(z)\Theta(\seps V(z/\eps)) + \seps \Big),
\end{align*}
and for which
\[
|\calR^\eps_2(z)| \leq \eps \, K \Big( \frac{|V(z/\eps)|}{\seps} + |V(z/\eps)|^2  + \seps \Big) \frac{1}{|\eps^{3/2} - 2i\omega_j(-1)^{l+1}\lambda_\eps(\kappa_j)|},
\]
after the change of variable $s\to z+\eps s$, and integrating in $s$. As a result, thanks to \eqref{eq:prop2V}, we have
\[
\sup_{z\in[0,L]}\E[|\calR^\eps_2(z)|]=\calO(\seps).
\]
To sum up, we obtain
\[
\sup_{z\in [0,L]} \E[|\calA^\eps(f_0^\eps + f_1^\eps + f_2^\eps + f_3^\eps)(z) - \calB^\eps_0(z)|] = \calO(\seps),
\]
and it only remains to determine the asymptotic of $\calB^\eps_0$.

\subsubsection{The term $\calB^\eps_0$}

This term needs a careful treatment to average the random process. In fact, for this term we cannot proceed as for $\calB^\eps_1$ with a test function like $f^\eps_2$ since there is no remaining phase. This would provide a term of the form $\int dp\,r(p)/\fg(p)=\infty$ for long-range correlations. The asymptotic of $\calB^\eps_0$ is given by the following result.

\begin{proposition}\label{prop:limB0}
We have
\[
\lim_{\eps \to 0}\sup_{z\in [0,L]}\E\Big[\Big| \int_0^z \calB^\eps_0(s)-\calL_M f(\calX_{\eps,M}(s))  ds\Big|\Big] = 0,
\] 
where $\calL_M$ is defined by \eqref{def:LM}.
\end{proposition} 

\begin{proof}[Proof of Proposition \ref{prop:limB0}]
Here, we only have to prove
\begin{equation}\label{eq:toprove}
\lim_{\eps \to 0}\sup_{z\in [0,L]}\E\Big[\Big| \int_0^z \calB^\eps_0(s)-\calL_{\eps,M} f(\calX_{\eps,M}(s))  ds\Big|\Big] = 0,
\end{equation}
with
\begin{align*}
\calL_{\eps,M} f(\calX) & :=   \sum_{j=1}^n\sum_{l=0,1} \theta'^2_0  \int dp\, \frac{r(p)}{\fg(p) + 2i\omega_j (-1)^{l+1}  \lambda_\eps(\kappa_j)}\\
&\times \Big( \bF_{1,j,j,l,1-l,\eps}(z) + \bF_{2,j,j,l,l,\eps}(z) + \bG_{1,j,l,\eps}(z)\Big),
\end{align*}
since passing from $\calL_{\eps,M}$ to $\calL_{M}=\calL_{\eps=0,M}$ being straightforward owing \eqref{eq:lambda_asymp}. Also, to simplify the notations in this proof, let us denote 
\begin{equation}\label{def:F}
\bF_{j,l,\eps}(z):= \frac{\bF_{1,j,j,l,1-l,\eps}(z)+\bF_{2,j,j,l,l,\eps}(z)+\bG_{1,j,l,\eps}(z)\big)}{\fg(p) + 2i\omega_j (-1)^{l+1}  \lambda_\eps(\kappa_j)},
\end{equation}
so that
\begin{align*}
\int_0^z \calB^\eps_0(s)-\calL_{\eps,M} f(\calX_{\eps,M}(s))  ds = \sum_{j=1}^n & \sum_{l=0,1} \theta'^2_0  \int_0^z ds\, \bF_{j,l,\eps}(s) \\
&\times \int V(s/\eps,dp)V(s/\eps,dq)-r(p)\delta(p-q)dpdq .
\end{align*}
To prove \eqref{eq:toprove}, we decompose the interval $[0,z]$ over a uniform grid with a small stepsize of order $\eps$ in order to get ride of the $s$-dependence of $\bF_{j,l,\eps}$, and then average out the $V$'s. Let $\eta>0$ be an arbitrary small parameter and write
\begin{align*}
\int_0^z \calB^\eps_0(s)-\calL_{\eps,M} f(\calX_{\eps,M}(s))  ds & = \sum_{m=0}^{[z/(\seps \eta)]-1} \int_{m\seps \eta}^{(m+1)\seps \eta} \calB^\eps_0(s)-\calL_{\eps,M} f(\calX_{\eps,M}(s))  ds \\
& + \int_{[z/(\seps \eta)]}^{z} \calB^\eps_0(s)-\calL_{\eps,M} f(\calX_{\eps,M}(s))  ds\\
&=: \calR^\eps_3(z) + \calR^\eps_4(z).
\end{align*}
For $\calR^\eps_4$, we have
\[
\E[|\calR^\eps_4(z)|]\leq (z-[z/(\seps \eta)] \seps \eta ) \Big(\sup_{s\in[0,L]}\E[|\calB^\eps_0(s)|]  + K \Big) \leq K' \seps,
\]
where  $K$ is some constant that bounds uniformly $\calL_{\eps,M} f(\calX_{\eps,M}(s))$. The last inequality is obtained by bounding $\sup_{\eps, s\in[0,L]}\E[|\calB^\eps_0(s)|]$ as a subcase of Lemma \ref{lem:bound_As}. 

Regarding $\calR^\eps_3$, we have
\begin{align*}
\calR^\eps_3(z)&=\sum_{j=1}^n\sum_{l=0,1} \theta'^2_0  \sum_{m=0}^{[z/(\seps \eta)]-1} \bF_{j,l,\eps}(mq\eta) \\
&\hspace{2cm}\times  \int_{m\seps \eta}^{(m+1)\seps \eta}ds\,\int \big( V(s/\eps,dp)V(s/\eps,dq)-r(p)\delta(p-q)dpdq \big)\\
& + \sum_{j=1}^n\sum_{l=0,1} \theta'^2_0  \sum_{m=0}^{[z/(\seps \eta)]-1} (\bF_{j,l,\eps}(s) - \bF_{j,l,\eps}(m \seps\eta) )\\
&\hspace{2cm}\times \int_{m\seps \eta}^{(m+1)\seps \eta}ds\,\int \big( V(s/\eps,dp)V(s/\eps,dq)-r(p)\delta(p-q)dpdq \big) \\
&=:\calR^\eps_{31}(z)+\calR^\eps_{32}(z).
\end{align*}
For $\calR^\eps_{32}(z)$, we have from \eqref{eq:calX_M2}
\[
\E[|\calR^\eps_{42}(z)|]\leq \frac{K}{\seps} \sum_{m=0}^{[z/(\seps \eta)]-1}  \int_{m\seps \eta}^{(m+1)\seps \eta}ds\,\int_{m\seps\eta}^s ds' \leq \eta\, K',
\] 
where the constant $K>0$ is obtained through similar arguments as to obtain the first part of Lemma \ref{bound:f2} using the denominator of \eqref{def:F}. Finally for $\calR^\eps_{31}(z)$, using the Cauchy-Schwarz inequality  w.r.t. the expectation, we have
\begin{align*}
\E [| \calR^\eps_{31}(z)| ] &\leq K  \sum_{m=0}^{[z/(\seps \eta)]-1}  \sqrt{I^\eps_{m,\eta}},
\end{align*}
for some constant $K>0$, with
\[
I^\eps_{m,\eta} := \E\Big[\Big|\int_{m\seps \eta}^{(m+1)\seps \eta} ds  \iint  V(s/\eps,dp) V(s/\eps,dq)-r(p)\delta(p-q)\,dp\,dq \Big|^2\Big].
\]
Therefore, using the Gaussianity of $V$, we have
\begin{align*}
I^\eps_{m,\eta} =\int_{m\seps \eta}^{(m+1)\seps \eta} & ds_1  \int_{m\seps \eta}^{(m+1)\seps \eta} ds_2 \\
& \times \iiiint \Big( \E\Big[V\Big(\frac{s_1}{\eps},dp_1\Big) V\Big(\frac{s_2}{\eps},dq_1\Big)\Big]\E\Big[V\Big(\frac{s_1}{\eps},dp_2\Big) V\Big(\frac{s_2}{\eps},dq_2\Big)\Big]\\
&+\E\Big[V\Big(\frac{s_1}{\eps},dp_1\Big) V\Big(\frac{s_2}{\eps},dq_2\Big)\Big]\E\Big[V\Big(\frac{s_1}{\eps},dp_2\Big), V\Big(\frac{s_2}{\eps},dq_1\Big)\Big]\Big)
\end{align*}
and by symmetry w.r.t. $s_1$ and $s_2$ we obtain
\[
I^\eps_{m,\eta} \leq \iint dp\,dq \,r(p)\, r(q)\int_{m\seps \eta}^{(m+1)\seps \eta} ds_1 \int_{m\seps \eta}^{s_1} ds_2 \, e^{-(\fg(p)+\fg(q))(s_1-s_2)/\eps}.
\]
Now, to bound $I^\eps_{m,\eta}$ properly, we consider two cases. Integrating without any caution would provide a term of the form $r(p)\, r(q)/(\fg(p)+\fg(q))$ which is not integrable at $0$ in case of long-range correlations. 

Let $\eta'>0$ be an arbitrary small parameter. If $|p|\leq \eta'$, we can bound the exponential term by $1$ and obtain a term of the form
\[
\eps\, \eta^2 \int_{\{|p|\leq \eta'\}} r(p) dp,
\]   
upto a constant independent of $m$, $\eps$, $\eta$ or $\eta'$, and where the term $\eps \eta^2$ is compensated by the one of the above sum in $m$. However, this term goes to $0$ as $\eta'\to 0$ thanks to the integrability of $r$ and the dominated convergence theorem. Now for $|p| > \eta'$, that is we place ourselves away from $0$, the point that can produce nonintegrability, we have
\begin{align*}
\int_{m\seps \eta}^{(m+1)\seps \eta} ds_1 \int_{m\seps \eta}^{s_1} ds_2 \, e^{-(\fg(p)+\fg(q))(s_1-s_2)/\eps} & \leq \frac{\eps}{\fg(p)}\int_0^{\seps \eta} (1-e^{-\fg(p)s_1/\eps })ds_1 \\
& \leq K_{\eta'} \eps^{3/2}\eta.
\end{align*} 
As a result, one can finally write
\[
\E [| \calR^\eps_{31}(z)| ] \leq K' \Big( \Big(\int_{\{|p|\leq \eta'\}} r(p) dp\Big)^{1/2} + K_{\eta'} \eps^{1/4}\Big),
\]
which concludes the proof of Proposition \ref{prop:limB0}.
\end{proof}

All the limit points are then solution to a martingale problem associated to the infinitesimal generator $\calL_M $ given by \eqref{def:LM}. From that definition it is direct to see that for $||\calX||\leq M$, with $\calX \in \mathbb{C}^{2n}$, we have
\[
\calL_M f(\calX) = \calL f(\calX),
\]
where
\begin{align*}
\calL f(\calX) := \calL f_{M=\infty}(\calX) = \sum_{j=1}^n\sum_{l=0,1} \frac{\theta'^2_0 \omega^2_j}{4c^2_0}  \int & \frac{r(p)}{\fg(p) + 2i(\omega_j/c_0) (-1)^{l+1}}\\
&\times \big( \calX^{j,1-l}\calX^{j,l} \, \partial^2_{\calX^{j,1-l}\calX^{j,l}} f(\calX)\\
& +  \calX^{j,1-l}_{\eps,M}\overline{\calX^{j,1-l}} \, \partial^2_{\overline{\calX^{j,l}}\calX^{j,l}} f(\calX ) \\
&+ \calX^{j,l} \partial_{\calX} \calX^{j,l}\partial_{\calX^{j,l}} f(\calX)\big)\\
& \hspace{-2cm} + c.c..
\end{align*}
Also, one can see that there is no coupling between two components $(\calX^{j,0}, \calX^{j',1})$ and $(\calX^{j,0}, \calX^{j',1})$ for $j\neq j'$, meaning that the components of the the limiting processes are independent. Therefore, we only have to look at the well-posedness of the (untruncated) martingale problem associated to one coordinate, that is for one frequency. The corresponding infinitesimal generator then writes
\begin{align*}
\bL f(\calX) & = \sum_{l=0,1} \frac{\theta'^2_0 \omega^2}{4c^2_0}  \int  \frac{r(p)}{\fg(p) + 2i(\omega/c_0) (-1)^{l+1}}\\
& \hspace{2cm}\times \Big( \calX^{1-l}\calX^{l} \, \partial^2_{\calX^{1-l}\calX^{l}} f(\calX) +  \calX^{1-l}\overline{\calX^{1-l}} \, \partial^2_{\overline{\calX^{l}}\calX^{l}} f(\calX ) + \calX^{l} \partial_{\calX^l} f(\calX)\Big)\\
& + c.c.\\
& =  \frac{\theta'^2_0 \omega^2 \Gamma_c(\omega)}{4 c^2_0}\Big(\calX^{1}\calX^{0} \, \partial^2_{\calX^{1}\calX^{0}} f(\calX)+\overline{\calX^{1}\calX^{0}} \, \partial^2_{\overline{\calX^{1}\calX^{0}}} f(\calX) \\
&\hspace{4cm}+\calX^{1}\overline{\calX^{1}} \, \partial^2_{\calX^{0}\overline{\calX^{0}}} f(\calX)+\calX^{0}\overline{\calX^{0}} \, \partial^2_{\calX^{1}\overline{\calX^{1}}} f(\calX) \Big) \\
&+\frac{\theta'^2_0 \omega^2 \Gamma_c(\omega)}{8 c^2_0}\Big(\calX^{0} \partial_{\calX^0} f(\calX)+\calX^{1} \partial_{\calX^1} f(\calX)+\overline{\calX^{0}} \partial_{\overline{\calX^0}} f(\calX)+\overline{\calX^{1}} \partial_{\overline{\calX^1}} f(\calX)\Big)\\
&+i\frac{\theta'^2_0 \omega^2 \Gamma_s(\omega)}{8 c^2_0}\Big(\calX^{0} \partial_{\calX^0} f(\calX) - \calX^{1} \partial_{\calX^1} f(\calX)-\overline{\calX^{0}} \partial_{\overline{\calX^0}} f(\calX)+\overline{\calX^{1}} \partial_{\overline{\calX^1}} f(\calX)\Big),
\end{align*}
for $\calX\in\mathbb{C}^2$,
where
\[
\Gamma_c(\omega):=2\int_0^\infty R(s)\cos(2 \omega s/c_0) ds\qquad\text{and}\qquad \Gamma_s(\omega):=2\int_0^\infty R(s)\sin(2 \omega s/c_0) ds.
\]
Finally, from a martingale representation theorem \cite[Proposition 4.6 pp. 315]{karatzas}, upto an extension of the underlying probability space, the solutions to this martingale problem can be represented as weak solutions to the following stochastic differential equation
\begin{align*}
d\calX_0(z) &= -\sqrt{\frac{\omega^2 \Gamma_c(\omega)}{ 4 c_0^2}}\begin{pmatrix} 0 & 1 \\ 1 & 0 \end{pmatrix}\calX_0(z) \circ dW_1(z) \\
&- i \sqrt{\frac{\omega^2 \Gamma_c(\omega)}{ 4 c_0^2}}\begin{pmatrix} 0 & 1 \\ -1 & 0 \end{pmatrix}\calX_0(z) \circ dW_2(z) \\
& - i\frac{\omega^2 \Gamma_s(\omega)}{8 c_0^2}\begin{pmatrix} 1 & 0 \\ 0 & 1 \end{pmatrix}\calX_0(z) dz,
\end{align*}
where $W_1$ and $W_2$ are two independent real-valued standard Brownian motions, and $\circ$ stands for the Stratonovich integral. This equation is readily well-posed, in the sense of probability law, since the diffusion coefficients and drift are linear \cite[Theorem 2.5 pp. 287 and Proposition 3.20 pp. 309]{karatzas}. As a result, all the accumulation points have the same distribution, and this concludes the proof of Theorem \ref{th:asymp}.

\section*{Conclusion}

In this paper, we have analyzed high-frequency waves propagating in randomly layered media with long-range correlations in the weak-coupling regime. In this context, the waves are affected in two ways. First, they exhibit a random travel time-shift that can be characterized as a fractional Brownian motion (or a Brownian motion for $\gamma=1$ in \eqref{eq:decay_R}), but with a standard deviation which is very large compared to the pulse width. Second, the wave-front spreading is deterministic and can be characterized as the solution to a paraxial wave equation involving a pseudo-differential operator. This operator exhibits a frequency-depend attenuation and phase modulation depending on the correlation function of the medium fluctuations, and ensuring the causality of the limiting paraxial equation as well as the Kramers-Kroning relations. The frequency-dependent attenuation is shown to be close to the form \eqref{eq:powerdecay} for $\lambda \in (1,2]$. Moreover, this pseudo-differential operator can be approximated by a fractional Weyl derivative, with order depending on $\gamma$, the power law decay of the correlation function of the random medium (see \eqref{eq:decay_R}).
 
The noise model considered in this paper presents some restrictions that could be removed since the scattering coefficients rely on the correlation function $R$, not the particular structure of $V$. Considering more general models could also allow us to obtain exponents $\lambda\in (0,1)$ for the attenuation power law, and the rough path theory could be helpful to handle more general settings. Nevertheless, our method opens the road to analyze more general 3D settings with random variations with respect to the transverse section, which will be the aim of future works. 

\appendix

\section{Uniqueness of \eqref{eq:parax1} and \eqref{eq:parax2}}\label{sec:uniqueness}

In this appendix, we only treat the uniqueness for \eqref{eq:parax1} since the methodology for \eqref{eq:parax2} is exactly the same.

From the linearity of \eqref{eq:parax1}, it enough to prove that any solution $\tilde{\mathcal{K}}$ to \eqref{eq:parax1} in $\calC^0_z([0,\infty),\calS'_{0,s,\by}(\R\times\R^2)) \cap \calC^1_z((0,\infty),\calS'_{0,s,\by}(\R\times\R^2))$, and vanishing at $z=0$, is constant equal to $0$. In other words, for any $Z>0$ and $\phi \in \mathcal{S}_0(\R\times \R^2)$, we need to prove that
\[\big<  \tilde{\mathcal{K}}(Z), \phi \big>_{\calS',\calS}=0.\]
To this end, let us consider in a first time $\psi \in \mathcal{S}(\R\times \R^2) $, and
\[
\varphi(z,s,\by) = \calK(\cdot,\cdot,z-Z) \ast \psi(s,\by) \qquad (z,s,\by)\in [0,Z]\times \R \times \R^2,
\] 
where $\calK$ is defined in the Fourier domain by \eqref{def:K}. Hence, $\varphi$ satisfies
\[
\partial^2_{sz} \varphi + \frac{c_0}{2} \Delta_{\by} \varphi - \calI(\varphi)=0,
\]
with $\varphi(z=Z)=\psi$. Using that $\varphi(z)\in \mathcal{S}(\R\times \R^2)$ for any $z\in[0,Z]$, we can consider
\[
g(z) = \big< \partial_s \tilde{\mathcal{K}}(z), \varphi(z) \big>_{\calS',\calS},
\]
which satisfies
\begin{align*}
\frac{d}{dz} g(z) & = \big< \partial^2_{sz} \tilde{\mathcal{K}}(z),  \varphi(z) \big>_{\calS',\calS} + \big< \partial_s \tilde{\mathcal{K}}, \partial_{z} \varphi(z) \big>_{\calS',\calS}\\
& = \big< (-c_0\Delta_{\by}/2 + \calI) \tilde{\mathcal{K}},  \varphi(z) \big>_{\calS',\calS} - \big< \tilde{\mathcal{K}}, \partial^2_{sz} \varphi(z) \big>_{\calS',\calS}\\
& = - \big< \tilde{\mathcal{K}}, ( \partial^2_{sz} + c_0\Delta_{\by}/2 - \calI)\varphi(z) \big>_{\calS',\calS}\\
& =0.
\end{align*}
The function $g$ being constant in $z$, we have
\[
\big<  \partial_s \tilde{\mathcal{K}}(Z), \psi \big>_{\calS',\calS} = g(Z)=g(0)=-\big<  \tilde{\mathcal{K}}(0), \partial_s\varphi(0)  \big>_{\calS',\calS} =0,
\]
so that $\partial_s \tilde{\mathcal{K}}(Z)=0$ in $\mathcal{S}(\R\times \R^2)$, and then $\tilde{\mathcal{K}}(s,\by,Z)=K(\by,Z)$ does not vary is $s$. As a result, for any $\phi \in \mathcal{S}_0(\R\times \R^2)$, we have
\[
\big<  \tilde{\mathcal{K}}(Z), \phi \big>_{\calS',\calS}=\int K(\by,Z)\Big(\int\phi(s,\by)ds\Big) d\by=0.
\]

\section{Proof of Proposition \ref{prop:KK}}\label{proof:KK}

The analyticity of $\omega\mapsto \omega^2 \Gamma_c(\omega)$ and $\omega\mapsto \omega^2 \Gamma_s(\omega)$ over the upper complex half-plane is direct. Let us now introduce a notation and make two remarks. Defining the inverse Fourier transform of \eqref{def:Fs} as
\[
\calF^{-1}(\psi)(s):=\frac{1}{2\pi}\int e^{-i\omega s} \psi(\omega) d\omega,
\]
the Fourier transform of the Hilbert transform reads
\[
\calF^{-1}(\calH(\psi))(s) = -i\,\text{sign}(s)\,\calF^{-1}(\psi)(s).
\]
From this relation, we see that $\calH(\calS(\R))\subset \calS(\R)$.
The second remark is that by applying two integrations by part we have
\[
\omega^2 (\Gamma_c(\omega) + i \, \Gamma_s(\omega)) = -\frac{c_0^2}{2}\int_0^\infty R''(s) e^{2i\omega s/c_0} ds,
\]
so that for any test function $\psi\in \calS(\R)$
\begin{align*}
\big< \calH(\omega'^2 \Gamma_c(\omega')), \psi \big>_{\calS',\calS} & = \frac{c_0^2}{2}\int_0^\infty ds\, R''(s) \int d\omega\, \cos(2\omega s/c_0)  \calH(\psi)(\omega) \\
& = - \frac{c_0^2}{2}\int_0^\infty ds\, R''(s) \int d\omega\, \calH(\cos(2\omega' s/c_0))(\omega)  \psi(\omega),
\end{align*}
thanks to the Fubini theorem since $R''$ is integrable. Now, using that
\[
\calH(\cos(2\omega' s/c_0))(\omega) = \sin(2\omega s/c_0),
\]
we obtain
\begin{align*}
\big< \calH(\omega'^2 \Gamma_c(\omega')), \psi \big>_{\calS',\calS} & = - \frac{c_0^2}{2}\int_0^\infty ds\, R''(s) \int d\omega\, \sin(2\omega s/c_0)  \psi(\omega) \\
&= \big< \omega^2 \Gamma_s(\omega), \psi \big>_{\calS',\calS}\,,
\end{align*}
yielding the first relation. To obtain the second relation, we follow the same lines but now using 
 \[
\calH(\sin(2\omega' s/c_0))(\omega) = -\cos(2\omega s/c_0),
\]
providing the minus sign for this relation, and then concludes the proof.

\section{Proof of Propositions \ref{prop:travel_time}}\label{proof:travel_time}

Setting
\[
\tilde W_\eps(L) := \frac{\seps \, \theta'_0}{\sigma_\eps}\int_0^L V(u/\eps)du,
\]
we have from a Taylor expansion of $\Theta$ at the second order (remember that $\Theta$ is odd, and then $\Theta''(0)=0$), and \eqref{eq:prop2V}
\begin{align*}
\E\big[|W_\eps(L)-\tilde W_\eps(L)|\big] & \leq \frac{1}{c_0 \sigma_\eps}\int_0^L \E \big[| \nu_\eps(u/\eps) - \seps \theta'_0 V(u/\eps) | \big] du \\
& \leq  \frac{\sup|\Theta'''|L}{c_0 \sigma_\eps} \eps^{3/2} \sup_{z\in [0,L/\eps]} \E\big[ |V(z)|^3 \big]\underset{\eps\to 0}{\longrightarrow 0},
\end{align*}
so that we only need to focus on $\tilde W_\eps$ according to \cite[Theorem 3.1 pp. 27]{billingsley}. For $\gamma >1$, the convergence of $\tilde W_\eps$ to $W_0$ is given by an invariant principle  (see \cite[Theorem 4]{marty1} for a more advanced results), and for $\gamma\in(0,1]$ we refer to \cite[Proposition 1.3]{gomez3}. This concludes the proof of the proposition.

\section{Proof of Proposition \ref{prop:delay}}\label{proof:delay}

The proof is in two steps. The first step consists in approximating the time delay in a more convenient form. The second step consists in looking at the expectation and variance of this new expression to obtain the convergence in probability. 

For the first step, using the Taylor expansion of $u\mapsto\sqrt{1+u}$ at the second order, we have
\[
\frac{\Delta T_\eps(L)}{\eps} = \frac{1}{8c_0 \eps}\int_0^L \nu^2_\eps(z/\eps)dz + E^1_\eps(L),
\]
with for any $\eta>0$
\[
\lim_{\eps\to0}\Pro\big(|E^1_\eps(L)|>\eta\big)=0.
\]
In fact, for any $\eta'>0$ we have
\begin{align*}
\Pro\big(|E^1_\eps(L)|>\eta\big) & \leq  \Pro\Big(|E^1_\eps(L)|>\eta, \seps \sup_{z\in[0,L/\eps]} |V(z/\eps)|\leq \eta' \Big) \\
& + \Pro\Big( \seps \sup_{z\in[0,L/\eps]} |V(z/\eps)| > \eta' \Big),
\end{align*}
where the second term on the r.h.s. goes to $0$ as $\eps \to 0$, thanks to \eqref{eq:prop1V} together with the Markov inequality. Now, working on the event $(\seps \sup_{z\in[0,L/\eps]} |V(z/\eps)|\leq \eta')$, we have for $\eta' \in (0,1/\sup|\Theta'|)$
\[
|E^1_\eps(L)| \leq \frac{\seps\sup|\Theta'|^3}{4c_0 (1-\eta' \sup|\Theta'|)^{5/2}} \int_0^L |V(z/\eps)|^3,
\]
so that
\begin{align*}
\E\Big[|E^1_\eps(L)|\1_{(\sup_{z\in[0,L/\eps]} |V(z/\eps)|\leq \eta')} \Big] & \leq \frac{\sup|\Theta'|^3 L}{4c_0 (1-\eta' \sup|\Theta'|)^{5/2}} \seps  \sup_{z\in[0,L/\eps]} \E[|V(z/\eps)|^3] \\
& \underset{\eps\to 0}{\longrightarrow 0},
\end{align*}
thanks to \eqref{eq:prop2V}. Also, we have from a second order Taylor expansion for $\Theta$ (using that $\Theta''(0)=0$)
\[
\frac{1}{8c_0 \eps }\int_0^L \nu^2_\eps(z/\eps)dz = D_\eps(L) + E^2_\eps(L),
\]
with 
\[
D_\eps(L) := \frac{\theta'^2_0}{8c_0 }\int_0^L V^2(z/\eps)dz,
\]
and for any $\eta>0$
\[
\lim_{\eps\to0}\Pro\big(|E^2_\eps(L)|>\eta\big)=0.
\]
In fact, we have
\begin{align*}
\E\Big[\Big|\frac{1}{8c_0 \eps }\int_0^L \nu^2_\eps(z/\eps)dz - D_\eps(L)\Big|\Big] & \leq \eps \frac{ |\theta'_0|\sup |\Theta'''| L}{24 c_0} \sup_{z\in[0,L/\eps]} \E[ |V(z)|^4]\\
& + \eps^3 \frac{ \sup|\Theta'''|^2 L}{288 c_0} \sup_{z\in[0,L/\eps]}  \E[|V(z)|^6],
\end{align*}
which goes to $0$ as $\eps\to 0$ thanks to \eqref{eq:prop2V}. As a result, from the Markov inequality, we have for any $\eta >0$
\[
\lim_{\eps\to 0} \Pro\Big(\Big|\frac{\Delta T_\eps(L)}{\eps}-D_\eps(L) \Big| > \eta \Big)=0,
\]
so that we only need to focus on the convergence of $D_\eps(L)$. Regarding its expectation we have
\[
\E[D_\eps(L)]=\frac{\theta'^2_0 R(0)L}{8c_0}.
\] 
Now, thanks to the Chebyshev inequality, for any $\eta>0$ we have 
\[
\Pro\Big( \Big|D_\eps(L)-\frac{\theta'^2_0 R(0)L}{8c_0}\Big| > \eta \Big) \leq \frac{Var[D_\eps(L)]}{\eta^2},
\]
and from the Gaussianity of $V$,
\begin{align*}
Var[D_\eps(L)] & =\frac{\theta'^4_0}{64 c_0^2} \int_0^L\int_0^L \E[V^2(z_1/\eps)V^2(z_2/\eps)] dz_1dz_2 - \Big(\frac{\theta'^2_0 R(0)L}{8c_0}\Big)^2 \\
& = \frac{\theta'^4_0}{32 c_0^2} \int_0^L\int_0^L R^2\Big(\frac{z_1 - z_2}{\eps}\Big) dz_1dz_2\\
& \underset{\eps \to 0}{\longrightarrow 0},
\end{align*}
which concludes the proof of the proposition thanks to the dominated convergence theorem.

\section{Proof of Lemma \ref{lem:Gamma_scatt}}\label{proof:lemma_coef_scat}

\paragraph{For $\gamma >1$}
Thanks to the integrability of the correlation function and making the change of variable $s\to l_0 s$, we have
\[
\Gamma_c(\omega,l_0)+i\, \Gamma_s(\omega,l_0) = 2\int_0^\infty R(s)e^{2i \omega s l_0/c_0}ds \underset{l_0\to 0}{\longrightarrow}2\int_0^\infty R(s)ds =  \Gamma_c(0),
\] 
using the dominated convergence theorem.

\paragraph{For $\gamma \in (0,1)$}
After integrating in $s$, we have
\[
\Gamma_c(\omega,l_0)+i\, \Gamma_s(\omega,l_0) = \frac{2}{l_0^\gamma} \int_S \frac{a(p)dp}{|p|^{2\alpha}(\mu|p|^{2\beta}/l_0 - 2i\omega/c_0)}. 
\]
Now, making the change of variable $p \to l_0^{1/(2\beta)}p$ gives
\begin{align*}
\Gamma_c(\omega,l_0)+i\, \Gamma_s(\omega,l_0) & = 2 \int_{S/l_0^{1/(2\beta)}} \frac{a(l_0^{1/(2\beta)} p)dp}{|p|^{2\alpha}(\mu|p|^{2\beta} - 2i\omega/c_0)}\\
&\underset{l_0\to 0}{\longrightarrow} 2a(0) \int_{-\infty}^\infty \frac{dp}{|p|^{2\alpha}(\mu|p|^{2\beta} - 2i\omega/c_0)}.
\end{align*}
This limit can be rewritten as 
\[
  a(0)\int_{-\infty}^\infty \frac{dp}{|p|^{2\alpha}(\mu|p|^{2\beta} - 2i\omega/c_0)} = \int_0^\infty \tilde R(s) e^{2 i\omega s}ds,
\]
with
\[
\tilde R(s) = a(0)\int_{-\infty}^\infty \frac{e^{-\mu |p|^{2\beta}s}}{|p|^{2\alpha}}dp = \frac{R_0}{s^\gamma}.
\]
\paragraph{For $\gamma=1$}
After integrating in $s$, we have
\begin{align*}
\Gamma_c(\omega,l_0)+i\, \Gamma_s(\omega,l_0) & = \frac{2}{l_0 |\ln(l_0)|} \int_{-r_S}^{r_S} \frac{a(p)dp}{|p|^{2\alpha}(\mu|p|^{2\beta}/l_0 - 2i\omega/c_0)}\\
&= \frac{4 }{l_0 |\ln(l_0)|} \int_{0}^{r_S} \frac{a(p)dp}{p^{2\alpha}(\mu \,  p^{2\beta}/l_0 - 2i\omega/c_0)} \\
& = I_{l_0} + II_{l_0} + III_{l_0}+IV_{l_0},
\end{align*}
where the decomposition is defined as follows:
\begin{align*}
I_{l_0} &:= \frac{4 }{l_0 |\ln(l_0)|} \int_{0}^{r_S} \frac{(a(p)-a(0))dp}{p^{2\alpha}(\mu \,  p^{2\beta}/l_0 - 2i\omega/c_0)},\\
II_{l_0} &:= \frac{4 a(0) }{l_0 |\ln(l_0)|} \int_{0}^{l_0^{1/(2\beta)}} \frac{dp}{p^{2\alpha}(\mu \,  p^{2\beta}/l_0 - 2i\omega/c_0)},\\
III_{l_0} &:= \frac{4 a(0) }{l_0 |\ln(l_0)|} \int_{l_0^{1/(2\beta)}}^{r_S} \frac{dp}{p^{2\alpha}}\Big( \frac{1}{\mu \,  p^{2\beta}/l_0 - 2i\omega/c_0 }-\frac{1}{\mu \,  p^{2\beta}/l_0}\Big),\\
IV_{l_0} &:= \frac{4 a(0) }{\mu |\ln(l_0)|} \int_{l_0^{1/(2\beta)}}^{r_S} \frac{dp}{p},
\end{align*}
using that $2(\alpha+\beta)=1$, since $\gamma=1$.

For $I_{l_0}$, using the triangular inequality and making the change of variable $p \to l_0^{1/(2\beta)}p$ gives
\[
|I_{l_0}|\leq \frac{4 r_S \sup|a'|}{\mu |\ln(l_0)|} \underset{l_0\to 0}{\longrightarrow} 0.
\]
For $II_{l_0}$, the idea being just to isolate the point $p=0$, we have
\[
|II_{l_0}|\leq \frac{4 a(0)c_0}{2\mu\, \omega_c |\ln(l_0)|} \int_{0}^{l_0^{1/(2\beta)}} \frac{dp}{p^{2\alpha}} \propto \frac{l_0}{|\ln(l_0)|} \underset{l_0\to 0}{\longrightarrow} 0.
\]
For $III_{l_0}$, we have after the change of variable $p \to l_0^{1/(2\beta)}p$
\[
|II_{l_0}|\leq \frac{4 a(0)2\omega_c}{\mu^2\, c_0 |\ln(l_0)|} \int_{1}^{r_S/l_0^{1/(2\beta)}} \frac{dp}{p^{2\alpha+4\beta}} \propto \frac{1}{|\ln(l_0)|} \underset{l_0\to 0}{\longrightarrow} 0.
\]
Finally, for $IV_{l_0}$, integrating in $p$ we have
\[
IV_{l_0} = \frac{4 a(0) }{\mu |\ln(l_0)|} \Big( \ln(r_S)-\frac{1}{2\beta}\ln(l_0)\Big)\underset{l_0\to 0}{\longrightarrow}  \frac{2 a(0) L'}{\mu \, \beta},
\]
which concludes the proof of Lemma \ref{lem:Gamma_scatt}.

\end{document}